\documentclass[a4paper,11pt]{article}

\usepackage{fullpage}
\usepackage{libertine}

\usepackage[normalem]{ulem}

\usepackage{amsmath,amssymb,amsthm}
\usepackage{enumitem}
\usepackage{booktabs}
\usepackage{algorithm}
\usepackage{algpseudocode}
\usepackage{float}
\usepackage{microtype}

\usepackage{multirow}

\usepackage[breaklinks=true]{hyperref}
\usepackage[svgnames]{xcolor}
\usepackage[capitalise,nameinlink]{cleveref}
\hypersetup{colorlinks={true},linkcolor={DarkBlue},citecolor=[named]{DarkGreen}}

\usepackage{natbib}

\newtheorem{theorem}{Theorem}[section]

\newtheorem{lemma}[theorem]{Lemma}
\newtheorem{claim}{Claim}

\newtheorem{corollary}[theorem]{Corollary}
\theoremstyle{definition}

\newcommand{\cost}{\text{cost}}
\newcommand{\Fone}{F_1}
\newcommand{\Ftwo}{F_2}
\newcommand{\SC}{\mathrm{SC}}
\newcommand{\OPT}{\mathrm{OPT}}
\newcommand{\dist}{d}
\newcommand{\One}{\{1\}}
\newcommand{\Two}{\{2\}}
\newcommand{\Both}{\{1,2\}}
\newcommand{\EPM}{\normalfont\textsc{Parity-Median}}
\newcommand{\Ppm}{\text{PM}}
\newcommand{\Mfourthirds}{M^{4/3}}
\newcommand{\SL}{\textsc{Serial-Lex}}
\newcommand{\OICM}{\textsc{Independent-Cut-Median}}
\newcommand{\bw}{\mathbf{w}}
\newcommand{\bo}{\mathbf{o}}

\usepackage{authblk}

\title{\bf Discrete Truthful Heterogeneous Two-Facility Location: The Line and Beyond}
\author[1]{Panagiotis Kanellopoulos}
\author[1,2]{Alexandros A. Voudouris}
\affil[1]{University of Essex, UK}
\affil[2]{University of Southern Denmark, Denmark}

\date{}

\begin{document}
\maketitle

\begin{abstract}
We study deterministic strategyproof mechanisms for discrete heterogeneous two-facility location. In our model, $n$ agents occupy distinct nodes of a connected graph and privately report non-empty approval preferences over two facilities, which must be placed at distinct nodes. The cost of an agent is her total distance from the facilities she approves, and the objective is to minimize the social cost (the total cost of the agents). For the line graph, the best possible approximation ratio of deterministic strategyproof mechanisms was previously shown to lie between $4/3$ and $17/4$. We close this gap by designing an optimal $4/3$-approximate mechanism that combines a fixed-parity median rule, which suffices for instances with $n\ge7$, and strategyproof local overrides for the remaining smaller cases. Beyond the line, we design a deterministic strategyproof $2$-approximate mechanism for every connected graph and prove a lower bound of $3/2$ on the graph $K_{1,3}$.
\end{abstract}

\section{Introduction}
\label{sec:intro}
Facility location problems model collective decision situations in which a planner must build public facilities while relying on information held by self-interested agents. Since agents may misreport if this can lead to a better outcome for them, the mechanism responsible for mapping reports to locations must be \emph{strategyproof}: no agent should ever be able to strictly decrease her cost by lying. However, strategyproofness is typically not well-aligned with efficiency, and the standard way to quantify this tension, initiated by \citet{procaccia09approximate}, is through the \emph{approximation ratio}: the worst-case ratio between the social cost of the outcome computed by the strategyproof mechanism and the optimal social cost achievable by any feasible outcome. Over the years, facility location has become one of the central benchmarks of {\em approximate mechanism design without money} and many different variants have been studied~\citep{fl-survey}.

In this work, we focus on a {\em discrete} version with agents that occupy distinct public nodes of a graph. There are two {\em heterogeneous} facilities that are not interchangeable; for example, think of a school and a library. The agents have private approval {\em preferences} over the facilities representing which of them they care about. So, in contrast to the classic homogeneous setting, the private information is not \emph{where} an agent is but \emph{what her preferences are}. Given the facility locations on the graph, each agent suffers an individual cost equal to her total graph distance from the facilities she approves, and the goal is to choose the locations to minimize the {\em social cost} (the total cost of the agents). This specific model was introduced by \citet{serafino-ventre} for the line graph, and was later studied by \citet{Kanellopoulos2023sidma} for the case of {\em active} agents with non-empty preferences. 
\citeauthor{Kanellopoulos2023sidma} proved that the best possible approximation ratio of deterministic strategyproof mechanisms lies between
$4/3$ and $17/4$, leaving open two natural questions: {\em What exactly is the best possible approximation for the line, and what happens beyond the line?} 
Graphs bring genuinely new phenomena, e.g., branching nodes or odd cycles, and it is not a priori clear whether the line bound persists, degrades, or is governed by identifiable structural features of the graph. In this paper, we resolve the first question completely and initiate the study of the second one. Our results are summarized in Table~\ref{tab:results}. 

\subsection{Our Contribution}
\label{sec:intro:contribution}

\begin{table}[t]
\centering
\begin{tabular}{lcc}
\toprule
Graph class & Lower bound & Upper bound \\
\midrule
Line & $4/3$~\citep{Kanellopoulos2023sidma} & $4/3$ (Thm.~\ref{thm:line:upper-bound-4/3}) \\
Connected graphs & $3/2$ (Thm.~\ref{thm:lower:3/2:general-graphs}) & $2$ (Thm.~\ref{thm:graph-upper}) \\
\bottomrule
\end{tabular}
\caption{Overview of our bounds on the approximation ratio of deterministic strategyproof mechanisms for the social cost. The lower bound of $3/2$ is shown on the claw graph $K_{1,3}$.}
\label{tab:results}
\end{table}

Our first main result is a strategyproof mechanism for the line graph with approximation ratio at most $4/3$ for every $n\ge4$, matching the lower bound of \citet{Kanellopoulos2023sidma}. For $n=3$, \citeauthor{Kanellopoulos2023sidma} had already provided a tight $4/3$ bound, while the {\sc TwoExtremes} mechanism of \citet{serafino-ventre} is optimal for $n\le2$.
Our mechanism $\Mfourthirds$ is best viewed as a single construction with case branches determined by the (public) number of agents:
\[
\Mfourthirds=
\begin{cases}
\EPM{} \text{ (Section~\ref{sec:line:parity})}, & n\ge7; \\
M_6 \text{ (Section~\ref{sec:line:n6})}, & n=6;\\
M_5 \text{ (Section~\ref{sec:line:n5})}, & n=5;\\
M_4 \text{ (Section~\ref{sec:line:n4})}, & n=4.
\end{cases}
\]
The conceptual core is the \EPM{} mechanism: it first partitions the nodes of the line into fixed opposite parity classes and associates each of them with a different facility. Then, it places each facility at the node of its parity class closest to the median of the agents that approve it. Fixing
disjoint domains guarantees feasibility (facilities are placed at different nodes) and makes strategyproofness a consequence of a simple median-monotonicity property. The parity structure ensures that the mechanism loses at most two additive units relative to the separable median optimum. Hence \EPM{} is $4/3$-approximate whenever $\OPT\ge6$, which always holds for $n\ge7$. The only genuine obstructions therefore occur for smaller instances with $n\in\{4,5,6\}$. For each of these cases, we present local overrides that are carefully designed to preserve strategyproofness and achieve the bound of $4/3$.

We next consider arbitrary connected graphs. For $n\le3$, a serial-lexicographic rule is strategyproof and $2$-approximate. For $n\ge4$, we introduce the \OICM{} mechanism. Starting from the nodes occupied by agents, it first constructs a maximal independent set of the subgraph induced by these nodes, and then extends it, using only empty nodes, to a maximal independent set $L$ of the whole graph; the complementary domain is $R=V\setminus L$, and note that both domains are non-empty and dominating. Facility $F_1$ is permanently restricted to $L$, facility $F_2$ to $R$, and then each facility is placed at a node of its side minimizing the total distance from the agents that approve it. Strategyproofness holds for exactly the same structural reason as on the line, that is, the domains are fixed and disjoint, and restricted medians are monotone under
the removal of a supporter. The approximation analysis uses the stronger structure of the construction: maximality inside the subgraph induced by the nodes occupied by agents supplies the missing unit in the only difficult case, where one facility has at most one supporter. This yields a deterministic strategyproof $2$-approximation for every connected graph.

On the negative side, we prove a lower bound of $3/2$ on the claw graph $K_{1,3}$ (Section~\ref{sec:general:lower}). This establishes the first separation between the line and more complex graphs for this model: the tight bound of $4/3$ on the line provably does not extend even to trees of maximum degree $3$. The proof is a combinatorial strategyproofness-propagation argument. We identify a small set of \emph{anchor} profiles such that any mechanism with approximation ratio better than $3/2$ is forced to choose a unique outcome,  and then follow chains of unilateral deviations that transfer these constraints to a terminal profile, where any admissible outcome creates a profitable manipulation.

\subsection{Related Work}
\label{sec:intro:related}
Facility location is a classic topic in social choice theory, motivated both by its practical importance and by its rich strategic structure. In the standard model, a single facility is placed on the line on the basis of reported single-peaked preferences; \citet{moulin1980} characterized the full class of strategyproof mechanisms for this setting. \citet{schummer-vohra2002} extended this line of work to networks, showing that the topology of the underlying graph places strong restrictions on the mechanisms that can satisfy strategyproofness. The approximation-theoretic study of facility location was initiated by \citet{procaccia09approximate} and has since developed into an extensive literature; see the survey of \citet{fl-survey}. Related work has also considered discrete domains and settings with multiple homogeneous facilities, including, e.g., characterizations on discrete lines and cycles~\citep{dokow-etal-2012} and tight deterministic approximation bounds for locating two facilities on the line~\citep{fotakis-tzamos-2014}.

The study of discrete heterogeneous facility location was initiated by \citet{serafino-ventre}. As already mentioned above, in their model, agents occupy distinct, publicly known nodes of a line and have private, possibly empty approval sets over two facilities. The facilities must be placed at distinct nodes, and each agent incurs a cost equal to the sum of her distances to the facilities she approves. We consider the slightly weaker model of \citet{Kanellopoulos2023sidma} in which every agent approves at least one facility. For this setting, we close the remaining approximation gap on the line by establishing a tight ratio of $4/3$ for every number of agents, and further extend the model to general connected graphs for which we establish constant approximation bounds.

Our work also contributes to the growing literature on \emph{feasibility-constrained facility location}. \citet{Sui2015constrained} studied how restrictions on admissible facility placements affect strategyproofness and efficiency, while \citet{Tang2020candidate} initiated the study of mechanisms for homogeneous facilities whose locations must be selected from a fixed set of candidate ones. The work most closely related to ours is that of \citet{kanellopoulos2025}, who considered two heterogeneous facilities and the same sum-of-distances individual cost objective in a candidate-location setting. Unlike in our model, however, the approval preferences of the agents are publicly known, whereas their positions are private. Related candidate-location models with the max-variant of individual cost were studied by \citet{Zhao2023constrained} and \citet{lotfi2024max}. Another form of joint feasibility constraint was considered by \citet{Xu2021minimum}, who required the selected facilities to satisfy a prescribed minimum-distance requirement. Candidate-location models have also been studied for obnoxious preferences~\citep{GaiLW22,ZhaoLNF24,Kanellopoulos2026obnoxious} and in the presence of predictions~\citep{FangFLNV2025predictions}. A different type of constraint appears in the work of \citet{Deligkas2025agent}, where facilities must be placed at distinct locations reported by the agents, making the feasible set itself dependent on their reports. Our model brings together several of these features: the nodes of the graph form a fixed discrete set of candidate locations, the two facilities must occupy distinct nodes, and agents strategically report which facilities they approve.

\section{Model and Preliminaries}
\label{sec:model}
Let $V=\{1,2,\ldots,m\}$ be the nodes of a connected {\em graph}, and denote by $\dist(u,v)$ the distance between any two nodes $u$ and $v$, defined as the length of the shortest path between $u$ and $v$. 
There is a set $N$ of $n$ {\em agents} at distinct public nodes $x_1, \ldots, x_n \in V$, and two {\em facilities} $\{\Fone,\Ftwo\}$.  
Each agent $i$ reports a non-empty {\em approval} set
$A_i\in\{\One,\Two,\Both\}$; let $A = (A_i)_{i \in N}$ denote the approval profile.
For $k\in\{1,2\}$, let $N_k=\{i:k\in A_i\}$ be the set of agents that approve facility $F_k$, and $n_k = |N_k|$. 
Although every agent approves at least one facility, it is possible that $N_k = \varnothing$ for some facility $F_k$ (which would imply that $N_{3-k}=N$). 

A feasible {\em outcome} is an ordered pair $(z_1,z_2)\in V^2$ with $z_1\ne z_2$, where $F_k$ is placed at $z_k$.  
Each agent $i$ suffers an individual cost equal to the total distance from all facilities she approves, that is,
\[
  \cost_i(z_1,z_2\mid A_i)=\sum_{k\in A_i} d(x_i,z_k).
\]
The {\em social cost} of an outcome $(z_1,z_2)$ is the cost of all agents, i.e., 
\[
  \SC(z_1,z_2\mid A)=\sum_{i=1}^n \cost_i(z_1,z_2\mid A_i)
  =\sum_{i\in N_1}d(x_i,z_1) + \sum_{i\in N_2} d(x_i,z_2).
\]
The optimal social cost is
\[
  \OPT(A)=\min_{z_1\ne z_2}\SC(z_1,z_2\mid A).
\]
A deterministic mechanism $M$ takes as input an approval profile $A$ and outputs a feasible outcome $M(A)$. It is \emph{strategyproof} if no agent can strictly decrease her true cost by changing her reported approval set while all other reports are fixed, that is, 
$\cost_i( M(A) \mid A_i) \leq \cost_i (M(A') \mid A_i)$ for any profile $A'$ that differs from $A$ only at the approval set of agent $i$. 
When clear from context, we will drop $A$ from notation, and simply write $\SC(z_1,z_2)$ for the social cost of a feasible outcome $(z_1,z_2)$, $\SC(M)$ for the social cost of the outcome computed by mechanism $M$, and $\OPT$ for the optimal social cost. 

We now present two useful lower bounds on the optimal social cost. 
For a non-empty set $S$ of agents, let $\mu(S)$ denote the set of all {\em median} nodes, which minimize the total distance from the agents of $S$. Since the optimal outcome must be feasible, we have the following lower bound on the social cost. 

\begin{lemma}
\label{lem:opt-median}
For any approval profile $A$ and $v_k \in \mu(N_k)$, $\OPT(A)\ge \sum_{k \in \{1,2\}}\sum_{i \in N_k} d(x_i,v_k)$. 
\end{lemma}

\noindent 
We will also make use of the next lower bound.

\begin{lemma}
\label{lem:opt-n-2-general}
For any approval profile $A$, $\OPT(A)\ge n_1+n_2 - 2 \ge n-2$.
\end{lemma}

\begin{proof}
For each $k\in\{1,2\}$, at most one agent in $N_k$ can have zero
contribution to the social cost from $F_k$. Since every non-zero
distance is at least $1$, the contribution of $F_k$ is at least
$n_k-1$. Hence, $\OPT\ge n_1+n_2-2\ge n-2$.
\end{proof}

\section{Line Graph: An Optimal $4/3$ Bound} 
\label{sec:line} 
In this section, we focus on the line graph and close the gap between $4/3$ and $17/4$ that was previously left open by \citet{Kanellopoulos2023sidma} for $n \geq 4$; note that \citeauthor{Kanellopoulos2023sidma} already showed a tight bound of $4/3$ for $n \leq 3$. In particular, we will show the following result. 

\begin{theorem}\label{thm:line:upper-bound-4/3}
    For the line graph and $n\geq 4$ agents, there is a deterministic strategyproof mechanism $M^{4/3}$ with approximation ratio at most $4/3$.  
\end{theorem}

\noindent 
Our mechanism works differently depending on the number of agents. Its main component is the {\sc Parity-Median} mechanism presented in Section~\ref{sec:line:parity} which appropriately places the two facilities at nodes of different parity and achieves the desired upper bound of $4/3$ when there are $n \geq 7$ agents. For instances with fewer agents, the mechanism overrides the {\sc Parity-Median} mechanism when certain special approval profiles are observed at input. These profiles are different depending on whether $n=6$ (Section~\ref{sec:line:n6}), $n=5$ (Section~\ref{sec:line:n5}), or $n=4$ (Section~\ref{sec:line:n4}). 

Before we present the mechanism and its analysis, we first provide a few additional preliminary definitions and remarks. Since the graph is a line, for any set of agents $S$, the set $\mu(S)$ of nodes that minimize the total distance to the agents of $S$ is a {\em median-interval}: if $|S|$ is odd, it contains the unique median location of $S$; if $|S|$ is even, $\mu(S)$ is the interval between the two median locations (that is, it contains the locations of the two median agents in $S$ and any empty nodes between them). For a node $z$ and an interval $I=[\ell,r]$, let
$d(z,I)=\min_{y\in I}d(z,y)$ be the distance of $z$ from $I$, defined as the distance of $z$ from its closest node in $I$; if $I =\varnothing$, let $d(z,I)=0$. 

Besides the two general lower bounds on the optimal social cost given as Lemma~\ref{lem:opt-median} and Lemma~\ref{lem:opt-n-2-general} in Section~\ref{sec:model}, we will also use the following, slightly stronger lower bound due to \citet{serafino-ventre} that is true for the line graph; see also \cite{Kanellopoulos2023sidma}. 

\begin{lemma}[\citep{serafino-ventre,Kanellopoulos2023sidma}]
\label{lem:opt-SV}
For any approval profile $A$, $\OPT(A)\ge \frac14 (n_1^2 + n_2^2-2)$.
\end{lemma}

\subsection{The {\sc Parity-Median} Mechanism and Instances with at Least Seven Agents} 
\label{sec:line:parity} 
In this section, we present the main component of our $4/3$-approximation bound, the \EPM{} ($\Ppm$) mechanism. Let $P_1$ and $P_2$ be two opposite parity classes of the line; that is, $P_1$ consists of the odd-numbered nodes, and $P_2$ consists of the even-numbered nodes. For each $k \in \{1,2\}$, facility $F_k$ will always be placed on a node in $P_k$, so the locations are always distinct. 
In particular, the mechanism works as follows: If $N_k\ne\varnothing$, then facility $F_k$ is placed at the leftmost node of parity $P_k$ closest to $\mu(N_k)$. Otherwise, if $N_k=\varnothing$, then $F_k$ is placed at an arbitrary fixed node of parity $P_k$. See Algorithm~\ref{alg:parity-median}.

\begin{algorithm}[H]
\caption{\EPM{}}
\label{alg:parity-median}
\begin{algorithmic}[1]
\Require Reported approval profile $A=(A_1,\ldots,A_n)$; fixed opposite parity classes $P_1,P_2$; fixed default nodes $q_k\in P_k$ for $k \in \{1,2\}$
\Ensure A feasible outcome $(z_1,z_2)$
\For{$k\in\{1,2\}$}
    \State $N_k\gets\{i:k\in A_i\}$
    \If{$N_k=\varnothing$}
        \State $z_k\gets q_k$
    \Else
        \State Compute the median interval $\mu(N_k)$ of the locations of agents in $N_k$
        \State $z_k\gets$ the leftmost node of parity $P_k$ closest to $\mu(N_k)$
    \EndIf
\EndFor
\State \Return $(z_1,z_2)$
\end{algorithmic}
\end{algorithm}

We will denote by $\Ppm_k(S)$ the location chosen for facility $F_k$ when the set of agents that approve it is $S$. 
We remark that an interval
$[\ell,r]$ with $\ell<r$ contains a node of $P_k$, while a singleton
$\{\ell\}$ contains one if and only if $\ell\in P_k$. Consequently, if
$\mu(S)=[\ell,r]$ contains a node of $P_k$, then
$\Ppm_k(S)\in\{\ell,\ell+1\}$; otherwise,
$\ell=r\notin P_k$ and $\Ppm_k(S)=\ell-1$, except when $\ell=1$, where
$\Ppm_k(S)=\ell+1$. In particular, $\ell-1\le\Ppm_k(S)\le\ell+1$, and
$\Ppm_k(S)\in\mu(S)$ whenever $\ell<r$. Moreover, the node
$\Ppm_k(S)$ minimizes the total distance from the agents of $S$ over the
nodes of $P_k$; that is, $\Ppm_k$ is a median rule restricted to $P_k$, with
a fixed tie-breaking rule.

The following lemma establishes that the distance of an agent from a facility she truly approves cannot decrease if the agent misreports by deleting her approval for that facility. This will be useful to show that \EPM{} is strategyproof.

\begin{lemma}
\label{lem:parity-sp}
Let $S$ be the set of agents that approve facility $F_k$. For any $i\in S$, 
\[d\left(x_i, \Ppm_k(S)\right) \le d\left(x_i, \Ppm_k(S\setminus\{i\})\right).\]
\end{lemma}

\begin{proof}
As observed above, $\Ppm_k(S)$ is a $P_k$-restricted median of $S$ with fixed
tie-breaking. The claim is therefore the special case $D=P_k$ of
Lemma~\ref{lem:cut-mono}, which we prove in Section~\ref{sec:graph-large} for
median rules restricted to an arbitrary fixed set of nodes $D$; its proof is
self-contained and does not use any material from this section.
\end{proof}

We can now show that \EPM{} is strategyproof. 

\begin{corollary}
\label{cor:epm-sp}
The \EPM{} mechanism is strategyproof.
\end{corollary}

\begin{proof}
The two parity classes are fixed, and thus no misreport can change the parity assignment of either facility or create a collision between them.  Moreover, the location of each facility depends only on the set of agents that approve it. Consider an agent and a facility she truly approves.  Deleting the approval for that facility cannot move it strictly closer to her by Lemma~\ref{lem:parity-sp}. Adding an approval for a facility she does not truly approve affects only a term outside her true cost. Since costs are additive over approved facilities, no deviation is profitable.
\end{proof}

The following lemma establishes an upper bound on the social cost of the outcome computed by \EPM{} compared to the optimum, and will be used repeatedly in the analysis of the mechanisms in the next sections.

\begin{lemma}
\label{lem:additive} 
For any $n\geq 4$, $\SC(\EPM{})\le  \OPT+2$.
\end{lemma}

\begin{proof}
If $|N_k|$ is even, then $\mu(N_k)$ contains both parity classes, and \EPM{} optimally places $F_k$ inside $\mu(N_k)$.  
On the other hand, if $|N_k|$ is odd, then $\mu(N_k)$ is a singleton. Moving one step to the prescribed parity, if necessary, increases the total distance by exactly one: non-median agents that approve $F_k$ can be paired on opposite sides of the median, and only the median pays the extra unit. Thus, by Lemma~\ref{lem:opt-median}, we have 
\begin{align*}
    \SC(\EPM{}) \leq \sum_{k \in \{1,2\}} \left( \sum_{i \in N_k} d(x_i,\mu(N_k)) + 1 \right) \leq \OPT + 2, 
\end{align*}
and the proof is complete.
\end{proof}

Using Lemma~\ref{lem:additive} and the lower bounds on the optimal social cost (Lemmas~\ref{lem:opt-SV} and~\ref{lem:opt-n-2-general}), we obtain the following result, which provides a first improvement for any $n \geq 4$, and settles the approximation for $n \geq 7$. 

\begin{theorem}
\label{thm:nge7}
The approximation ratio of {\normalfont \EPM{}} is at most $2$ for any $n \ge 4$, and at most $4/3$ for any $n \ge 7$.
\end{theorem}

\begin{proof}
For $n \geq 4$, by Lemma~\ref{lem:opt-n-2-general}, we have that $\OPT \geq 2$. Consequently, Lemma~\ref{lem:additive} implies that $\SC(\EPM{}\mid A)\le \OPT+2 \leq 2 \cdot \OPT$.

When $n \geq 7$, using Lemma~\ref{lem:opt-SV}, together with $n_1+n_2\geq n$, we obtain that $\OPT\geq 6$. In particular, observe that the minimum of $n_1^2+n_2^2$ occurs when $n_1$ and $n_2$ are as balanced as possible; so, when $n_1+n_2\geq n \geq 7$, this happens when $n_1=4$ and $n_2=3$ up to symmetry. Therefore, $n_1^2+n_2^2 \geq 4^2+3^2 = 25$, and Lemma~\ref{lem:opt-SV} gives us that $\OPT \geq 23/4 = 5.75$. Since the social cost is always an integer, this implies that it must be $\OPT \geq 6$. Given this, Lemma~\ref{lem:additive} implies that $\SC(\EPM{}\mid A)\le \OPT+2 \leq \frac43 \cdot \OPT$.
\end{proof}

\subsection{Instances with Six Agents}
\label{sec:line:n6}
By Lemma~\ref{lem:additive}, \EPM{} leads to a $4/3$-approximation whenever the optimal social cost is at least $6$, which is always true for instances with at least $7$ agents (see Theorem~\ref{thm:nge7}). For $6$ agents, the only instances where the optimal social cost is smaller than $6$ and \EPM{} might make non-optimal decisions are the ones consisting of a {\em contiguous singleton split}, up to relabeling: one facility is approved by the first three agents, and the other facility is approved by the remaining three agents. 
Here, we introduce mechanism $M_6$ which overrides \EPM{} for instances with this structure, where the mechanism anchors the left block at its middle agent and keeps the other facility under the parity-median rule. To preserve strategyproofness, this override is not done only for low-optimum instances with a contiguous singleton split structure, but for a slightly broader class of instances. See Algorithm~\ref{alg:m6}.

\begin{algorithm}[H]
\caption{The six-agent mechanism $M_6$}
\label{alg:m6}
\begin{algorithmic}[1]
\Require Reported approval profile $A=(A_1,\ldots,A_6)$
\Ensure A feasible outcome $(z_1,z_2)$
\State Compute $N_1$ and $N_2$
\ForAll{$(a,b)\in\{(1,2),(2,1)\}$}
    \If{$a\in A_i$ for $i\in\{1,2,3\}$ and $A_i=\{b\}$ for $i\in\{4,5,6\}$}
        \State $z_a\gets x_2$\label{line:m6-correct-a}
        \State $z_b\gets \Ppm_b(N_b)$\label{line:m6-correct-b}
        \State \Return $(z_1,z_2)$\label{line:m6-return}
    \EndIf
\EndFor
\State \Return \EPM{}$(A)$\label{line:m6-fallback}
\end{algorithmic}
\end{algorithm}

We first argue that the outcome computed by $M_6$ is well-defined, that is, the two facilities are placed at different nodes. 

\begin{lemma}
\label{lem:m6-well-defined}
Mechanism $M_6$ outputs a feasible outcome. 
\end{lemma}

\begin{proof}
The outcome is clearly feasible in the case where the output of $M_6$ coincides with that of \EPM{}. Consider now the special case where $a\in A_i$ for $i\in\{1,2,3\}$ and $A_i=\{b\}$ for $i\in\{4,5,6\}$. Here, $N_b$ contains agents $4, 5, 6$ and possibly some of the agents $1,2,3$. Hence, the median interval $\mu(N_b)$ of $N_b$ is $\{x_5\}$, $[x_4,x_5]$, $\{x_4\}$, or $[x_3,x_4]$, according to whether $|N_b|$ is $3$, $4$, $5$, or $6$. In any case $\Ppm_b(N_b)\ge x_3>x_2$, and thus the locations of the two facilities are distinct.
\end{proof}

We now argue that $M_6$ is strategyproof and then prove it achieves an approximation ratio of $4/3$. 

\begin{theorem}
\label{lem:m6-sp}
Mechanism $M_6$ is strategyproof.
\end{theorem}

\begin{proof}
Due to symmetry, it suffices to consider the case where $(a,b)=(1,2)$. Then, the override happens when we have an instance with a special structure such that $1 \in A_1, A_2, A_3$, and $A_4 = A_5 = A_6 = \Two$. The outcome for such an instance is to place $F_1$ at $x_2$ and $F_2$ at $\Ppm_2(N_2)$.
First, suppose that such an instance is the truthful one, and consider the possible deviations:
\begin{itemize}
    \item Agents $\{1,2,3\}$. Recall that $1 \in A_1, A_2, A_3$. 
If such an agent $i$ changes only her approval of $F_2$ (adding if $2 \not\in A_i$, or deleting if $2 \in A_i$), there is no change in the location of $F_1$ which is still placed at $x_2$. If $i$ truly approves $F_2$, then deleting her approval of $F_2$ cannot move $F_2$ closer by Lemma~\ref{lem:parity-sp}.  
So, suppose instead that one of them drops approval for $F_1$ and reports $\Two$.
This leads to an instance without the special structure that is handled by \EPM{}. 
Agent $2$ gains nothing this way, as her $F_1$-cost in the special instance was already zero. 
If the deviator is agent $1$, then $\mu(N_1 \setminus\{1\})=[x_2,x_3]$; this interval contains nodes of both parities, and thus the new location of $F_1$ lies in $[x_2,x_3]$, no closer to $x_1$ than $x_2$ was.  
If the deviator is agent $3$, then, similarly, $F_1$ is placed in $\mu(N_1 \setminus\{3\}) = [x_1,x_2]$, no closer to $x_3$ than before. 
The $F_2$-term of $i$'s true cost is also unaffected. A deviator who truly approves $F_2$ must have reported $\Both$, so switching to $\Two$ leaves $N_2$ unchanged and thus $F_2$ is again placed at $\Ppm_2(N_2)$. Finally, a deviator who does not approve $F_2$ is indifferent to its location.

\item
Agents $\{4,5,6\}$. Recall that $A_4 = A_5 = A_6 = \Two$. 
If such an agent $i$ changes to reporting $\Both$, then the instance changes into one that is handled by \EPM{}, but since $N_2$ remains unchanged, $F_2$ is again placed at $\Ppm_2(N_2)$, and the cost of $i$ does not decrease. 
If instead $i$ changes to reporting $\One$, then $i$ is deleted from $N_2$, and Lemma~\ref{lem:parity-sp} prevents $F_2$ from moving closer to $i$. 
Therefore, in any case, no agent in the right block can benefit. 
\end{itemize}

Now consider an instance without the special structure (that is handled by \EPM{}) such that an agent deviates and leads to an instance with the special structure. We again consider the possible deviations:
\begin{itemize}
    \item Agents $\{1,2,3\}$. For a deviation of such an agent $i$ to lead to the special instance, it must be the case that $1 \not\in A_i$, and thus $A_i = \Two$. If $i$ changes to reporting $\Both$, then $N_2$ remains unchanged and $F_2$ is again placed at $\Ppm_2(N_2)$.
    If $i$ changes to reporting $\One$, then $i$ is deleted from $N_2$, and Lemma~\ref{lem:parity-sp} prevents her truly-approved facility $F_2$ from moving closer. So, there is no profitable deviation for $i$.

    \item Agents $\{4,5,6\}$. For a deviation of such an agent $i$ to lead to the special instance, it must be the case that $i$ misreports her approval set as $\Two$. If it is truly the case that $1 \in A_i$, then $N_1=\{1,2,3,i\}$ and, since $\mu(N_1) = [x_2,x_3]$, the \EPM{} location of $F_1$ lies in $[x_2,x_3]$. After the deviation, the override places $F_1$ at $x_2$, which is weakly farther from $i$.  If $i$ also truly approves $F_2$, then $N_2$ remains unchanged after the deviation, and the location of $F_2$ remains $\Ppm_2(N_2)$. So, overall, $i$ cannot decrease her cost. 
\end{itemize}

All remaining deviations are between instances that are handled by the \EPM{} mechanism, and are thus not profitable by Lemma~\ref{lem:parity-sp}. Therefore $M_6$ is strategyproof.
\end{proof}

\begin{theorem}
\label{thm:n6-four-thirds}
The approximation ratio of $M_6$ is at most $4/3$.
\end{theorem}

\begin{proof}
We consider each of the two cases of $M_6$ separately. Suppose first that the input instance has the special structure that leads to $(z_a = x_2, z_b = \Ppm_b(N_b))$. Then, facility $F_a$ is placed at $x_2$, the median of the left block $\{1,2,3\}$. Since none of the three rightmost agents approves $F_a$, $F_a$ is placed at an optimal location. Facility $F_b$ is placed according to the \EPM{} mechanism, and thus at most one unit is lost relative to its one-facility optimum. Overall, we have that $\SC(M_6)\le \OPT+1$. By Lemma~\ref{lem:opt-n-2-general}, for $n\geq 6$, we also have that $\OPT\ge4$. Together, these imply that $\SC(M_6) \le(4/3)\OPT$.

Next, consider any other instance where $M_6$ outputs exactly the same outcome as the \EPM{} mechanism.
By Lemma~\ref{lem:opt-SV}, we have that if $n_1+n_2\geq 7$, then $\OPT\geq 6$ and then, by Lemma~\ref{lem:additive}, we get that the ratio is at most $4/3$.  This leaves us with singleton-only approval profiles, where $n_1+n_2=6$, which we consider below.

If $n_1$ and $n_2$ are even, then \EPM{} is in fact optimal as it places each facility at median nodes. If $(n_1, n_2)\in \{(1, 5), (5,1)\}$, then again by Lemma~\ref{lem:opt-SV} we get $\OPT\geq 6$ and, consequently, a ratio of at most $4/3$; the only possible approval profile, then, is $(n_1, n_2) = (3,3)$. 
If the two triples are not contiguous in the order of agents, then the sum of their one-facility median costs is at least $6$, and thus $\OPT\ge6$, leading to an approximation ratio of at most $4/3$. Hence the only possible split is contiguous: $\{1,2,3\}$ against $\{4,5,6\}$, up to swapping facility labels. But this is exactly the special contiguous singleton split structure, contradicting that the mechanism works exactly as the \EPM{} mechanism. 
\end{proof}

\subsection{Instances with Five Agents}
\label{sec:line:n5}
For five agents, the only possible obstruction to the $4/3$ ratio for \EPM{} occurs when there are four singleton-approval agents and one agent that approves both facilities; this means that $n_1+n_2=6$. Recall that, when $n_1+n_2\geq 7$, we have already established that \EPM{} is $4/3$-approximate. When $n_1+n_2=5$, we have that either $n_1$ or $n_2$ is even, and \EPM{} suffers at most one parity loss; thus, since $\OPT\geq 3$ (due to Lemma~\ref{lem:opt-n-2-general}), it is again $4/3$-approximate.

Mechanism $M_5$ handles a broader class of approval profiles carefully, and otherwise coincides with \EPM{}. In particular, $M_5$ handles differently two symmetric cases: 
\begin{itemize}
\item  If $A_1=A_2=\One$ and $2 \in A_4,A_5$, it outputs \(M_5(A)=\bigl(\Ppm_1(N_1),x_4\bigr)\); 
\item If $A_1=A_2=\Two$ and $1 \in A_4, A_5$, it outputs \(M_5(A)=\bigl(x_4,\Ppm_2(N_2)\bigr)\).
\end{itemize}
See Algorithm~\ref{alg:m5}.

\begin{algorithm}[H]
\caption{The five-agent mechanism $M_5$}
\label{alg:m5}
\begin{algorithmic}[1]
\Require Reported approval profile $A=(A_1,\ldots,A_5)$
\Ensure A feasible outcome $(z_1,z_2)$
\State Compute $N_1$ and $N_2$
\If{$A_1=A_2=\One$ and $2\in A_4$ and $2\in A_5$}
    \State \Return $(\Ppm_1(N_1),x_4)$\label{line:m5-first}
\EndIf
\If{$A_1=A_2=\Two$ and $1\in A_4$ and $1\in A_5$}
    \State \Return $(x_4,\Ppm_2(N_2))$\label{line:m5-second}
\EndIf
\State \Return  \EPM{}$(A)$\label{line:m5-fallback}
\end{algorithmic}
\end{algorithm}

We first argue about the soundness of $M_5$ in terms of computing a feasible outcome in which the two facilities are placed at different nodes of the line. 

\begin{lemma}
\label{lem:m5-well-defined} 
Mechanism \(M_5\) outputs a feasible outcome.
\end{lemma}

\begin{proof}
First, note that the two override branches are mutually exclusive, since the first requires
$A_1=A_2=\One$ and the second requires $A_1=A_2=\Two$. Due to their symmetry, it suffices to consider the first branch, with output
$(z_1,z_2)=\bigl(\Ppm_1(N_1),x_4\bigr)$ and show that $\Ppm_1(N_1)\ne x_4$. 
Let $\ell$ and $r$ be nodes such that $\mu(N_1)=[\ell,r]$, and recall how the leftmost node of parity $P_1$ closest to $[\ell,r]$ is determined: it is $\ell$ or $\ell+1$ when $[\ell,r]$ contains a node of $P_1$; otherwise it is $\ell-1$, except when $\ell=1$, where it is $\ell+1$. In every case, $\Ppm_1(N_1)\le\ell+1$.

In the branch we consider, the two leftmost locations of agents in $N_1$ are $x_1$ and $x_2$. 
If $|N_1|\le4$, then $\ell$ is the first or the second smallest location of agents in $N_1$, that is, $\ell\le x_2$, which implies that $\Ppm_1(N_1)\le x_2+1\le x_3<x_4$. 
If $|N_1|=5$, then $\mu(N_1)=\{x_3\}$ and, since $x_3$ is at least the third node of the line, the selected node is $x_3$ or $x_3-1$, again smaller than $x_4$.
\end{proof}

We now move to proving that $M_5$ is strategyproof and, then, we focus on its approximation ratio.

\begin{lemma}
\label{lem:m5-sp}
Mechanism \(M_5\) is strategyproof.
\end{lemma}
\begin{proof}
Let $B_1$ denote the first override branch ($A_1=A_2=\One$, $2\in A_4$, $2\in A_5$) with outcome $(\Ppm_1(N_1),x_4)$, and let $B_2$ denote the symmetric second branch with outcome $(x_4,\Ppm_2(N_2))$. Since $B_1$ requires agents $1$ and $2$ to report $\One$ while $B_2$ requires them to report $\Two$, no unilateral deviation can lead a profile in $B_1$ to a profile in $B_2$, or vice versa. In addition, since $\EPM{}$ is strategyproof, deviations between two fallback profiles handled exclusively by $\EPM{}$ are not profitable. By symmetry, it therefore suffices to consider deviations between two profiles in $B_1$, or between a profile in $B_1$ and one in the fallback branch.

Let $i$ be the deviating agent. 
The key observation is that for every profile in $B_1$ or the fallback branch, the mechanism places $F_1$ at $\Ppm_1(N_1)$.  
Hence, the $F_1$-term of $i$'s true cost behaves exactly as under \EPM{}. 
If $i$ truly approves $F_1$, her truthful report contains $1$. 
So, any deviation of $i$ either does not change $N_1$ and the location of $F_1$, 
or leads to $i$ being removed from $N_1$, in which case $F_1$ cannot move strictly closer to her (Lemma~\ref{lem:parity-sp}). 
If $i$ does not truly approve $F_1$, its location is irrelevant. 
So, it remains to show that no deviation improves the $F_2$-term of $i$'s true cost in the various cases.
\begin{itemize}
\item Both profiles involved in the deviation lie in $B_1$. 
Then $F_2$ is placed at $x_4$ for both profiles, and thus the $F_2$-term of $i$'s true cost is unchanged.

\item 
The truthful profile lies in $B_1$ and the deviation leads to the fallback branch.  
If $i$ is agent $1$ or agent $2$, her true type is $\One$, so $F_2$ does not affect her cost.  
Since agent $3$ cannot affect whether $B_1$ applies, the only other possibility is that agent $4$ or agent $5$ drops her approval of $F_2$ and reports $\One$.  Agent $4$'s truthful $F_2$-cost is $d(x_4,x_4)=0$ and thus cannot improve.  
For agent $5$, since $\mu(N_2 \setminus \{5\})$ is either $\{x_4\}$ or $[x_3,x_4]$, the new location of $F_2$ lies weakly to the left of $x_4$ and is no closer to $x_5$ than $x_4$ was.

\item 
The truthful profile lies in the fallback branch and the deviation creates $B_1$.  
If $i$ is agent $4$ or agent $5$, her truthful approval does not contain $2$, and thus $i$'s true type is $\One$, which implies that the location of $F_2$ does not affect her true cost. As in the previous case, agent $3$ cannot affect whether $B_1$ applies.  
If $i$ is agent $1$ or agent $2$, she deviates to $\One$, so her true type must contain $2$. 
Since $\mu(N_2)$ is either $\{x_4\}$ or $[x_3,x_4]$, we have that $\Ppm_2(N_2)\in[x_3,x_4]$.  
After the deviation, $F_2$ is placed at $x_4$, which is not closer to $i$'s location than $\Ppm_2(N_2)$. 
\end{itemize}
In all cases neither term of $i$'s true cost can be improved. Therefore $M_5$ is strategyproof.
\end{proof}

To bound the approximation ratio of $M_5$, it will be useful to argue about approval profiles where $n_1=n_2=3$. Note that, for five agents, this means that the two approval sets overlap exactly once. 

\begin{lemma}
\label{lem:low-span-triples}
Consider an approval profile over five agents with $n_1 = n_2 = 3$. If the sum of the index spans is at most five, then, up to swapping the two triples, the pair is one of 
\( (\{1,2,3\},\{3,4,5\}),(\{1,2,3\},\{2,4,5\}),\)
or
\(
  (\{1,2,4\},\{3,4,5\}).
\)
\end{lemma}
\begin{proof}
Note that any triple of distinct indices has span at least $2$, and this occurs only if the indices are consecutive numbers.  Furthermore, any triple where two successive entries differ by more than $2$ has a span of at least $4$. So, we can have at most one triple with entries not being consecutive numbers, and, in that case, the single gap between successive entries should be exactly $2$. Given that the two triples cover all five indices and overlap once, up to swapping, either both triples are consecutive, forcing $(\{1,2,3\}$,$\{3,4,5\})$; or exactly one has a single gap of $2$, which (again up to swapping) is $(\{1,2,3\}$, $\{2,4,5\})$ or $(\{1,2,4\}$, $\{3,4,5\})$.
\end{proof}

We are now ready to prove the upper bound of $4/3$ on the approximation ratio. 

\begin{theorem}
\label{thm:n5-four-thirds}
The approximation ratio of $M_5$ is at most $4/3$.
\end{theorem}

\begin{proof}
We will argue about the two possible branches of $M_5$ separately.

\medskip
\noindent
{\bf No override branch of $M_5$ applies.}
Then $M_5$ coincides with \EPM{}. 
If the outcome suffers at most one parity loss, then, since \(\OPT\ge3\) (due to Lemma~\ref{lem:opt-n-2-general} for $n=5$), 
\( \SC(M_5)\le \OPT+1\le \frac{4}{3}\cdot \OPT \). So, it remains to consider the case where both facilities suffer a parity loss, which implies that both approval sets are non-empty and of odd cardinality. We will show that, in this case, $\OPT\ge6$, and Lemma~\ref{lem:additive} will then give us that $\SC(M_5)\le\OPT+2\le\frac43\cdot\OPT$. Observe that this is clearly true if there exists a $k \in \{1,2\}$ such that $|N_k| = 5$; indeed, since $N_1,N_2 \neq \varnothing$, we have that $n_1^2+n_2^2\ge 5^2+1^2=26$, and Lemma~\ref{lem:opt-SV} implies $\OPT\ge6$. Therefore, since $n_1+n_2\geq 5$, the only possible remaining case is when $|N_1|=|N_2|=3$. 

For a triple \(T=\{i<j<k\}\), its one-facility optimal social cost is at least $(x_k - x_j) + (x_j - x_i) = x_k-x_i\ge k-i=\operatorname{span}(T)$. Hence, if \(\operatorname{span}(N_1)+\operatorname{span}(N_2)\ge6\), then \(\OPT\ge6\). So, we can now focus on the case where the  total index span is at most \(5\).  By Lemma~\ref{lem:low-span-triples}, up to swapping
facilities, only three pairs remain:
\(
(\{1,2,3\},\{3,4,5\})\), \((
\{1,2,3\},\{2,4,5\})\), and \((
\{1,2,4\},\{3,4,5\}).
\)
We consider each of them:
\begin{itemize}
\item \((\{1,2,3\},\{2,4,5\})\). 
Since $\mu(N_1) = x_2$ and $\mu(N_2) = x_4$, we have that $\OPT \geq (x_3-x_1)+(x_5-x_2)$. 
If $(x_3-x_1)+(x_5-x_2) = 5$, then all four consecutive gaps are \(1\), the two medians \(x_2\) and \(x_4\) have the
same parity, and thus it cannot be the case that both facilities suffer parity losses. Therefore, $\OPT \geq 6$. 
A similar argument applies to the pair
\((\{1,2,4\},\{3,4,5\})\).

\item 
$(N_a=\{1,2,3\}, N_b=\{3,4,5\})$ for $\{a,b\}=\{1,2\}$. Since $A_1=A_2 = \{a\}$ and $b \in A_4, A_5$, an override branch of $M_5$ applies, either the first in case $a=1$ or the second in case $a=2$, contradicting the assumption that no override branch applies. 
\end{itemize}

\medskip
\noindent
{\bf An override branch of $M_5$ applies.}
Due to symmetry, it suffices to consider the first override branch, where the mechanism outputs \((z_1, z_2)  = (\Ppm_1(N_1), x_4).\)
Here, $N_2$ consists of agents $4$, $5$, and possibly of agent $3$ as well, that is, 
\(N_2\in\bigl\{\{4,5\},\{3,4,5\}\bigr\}.
\)
If $N_2=\{4,5\}$, the median interval is $[x_4,x_5]$; if $N_2=\{3,4,5\}$, the median is $x_4$. 
Thus, $F_2$ is always placed at a median agent (agent $4$) and its contribution to the social cost is at most as much as its contribution to the optimal social cost. Facility $F_1$ is placed according to $\EPM{}$, and thus
its contribution is at most one more than its contribution to the optimal social cost.
Overall, we have that $\SC(M_5)\le \OPT+1 \le \frac43 \cdot \OPT$, where the last inequality follows from Lemma~\ref{lem:opt-n-2-general} for $n=5$ which gives us that $\OPT \geq 3$.
\end{proof}

\subsection{Instances with Four Agents}
\label{sec:line:n4}
Contrary to the mechanisms in the previous sections, where deviations from \EPM{} depended on the approval sets of the agents, for instances with $n=4$ agents, such deviations depend only on the gap between agents $2$ and $3$. The hard instances are those with a small middle gap:
when the two central agents are close, both facilities are pulled towards the same central node, and a single parity correction is not enough. We therefore switch, on the basis of this public middle gap, between a tailored small-gap rule and \EPM{}.

Let $x_1<x_2<x_3<x_4$ be the four public locations and set $g=x_3-x_2$. 
The four-agent mechanism $M_4$ branches on $g$: 
if $g\le4$ it runs a small-gap rule $M_{\le4}$, and if $g\ge5$ it runs $\EPM{}$. 
Rule $M_{\le4}$ appropriately applies the following three  branches, where $O_L=(x_2,x_3)$ places
$\Fone$ at $x_2$ and $\Ftwo$ at $x_3$, and $O_R=(x_3,x_2)$ swaps them.
\begin{enumerate}[label=(\roman*),leftmargin=3em]
\item If $N_k=\{1\}$ for some $k\in\{1,2\}$, output $z_k=x_1$ and $z_{3-k}=x_3$.
\item If $N_k=\{4\}$ for some $k\in\{1,2\}$, output $z_k=x_4$ and $z_{3-k}=x_2$.
\item Otherwise, output the smaller-cost outcome between $O_L$ and $O_R$, breaking ties in favor of $O_L$.
\end{enumerate}
Clearly, at most one of these endpoint branches can apply: if $N_k=\{1\}$, then agents $2,3,4$ all approve $F_{3-k}$, and thus $F_{3-k}$ cannot have singleton support $\{4\}$. Consequently, the computed outcome is feasible. See Algorithm~\ref{alg:m4} for a full description.

\begin{algorithm}[H]
\caption{The four-agent mechanism $M_4$}
\label{alg:m4}
\begin{algorithmic}[1]
\Require Reported approval profile $A=(A_1,\ldots,A_4)$
\Ensure A feasible outcome $(z_1,z_2)$
\State Compute $N_1$ and $N_2$
\State $g\gets x_3-x_2$
\If{$g\le4$}\label{line:m4-gap} 
    \For{each $k\in\{1,2\}$}
        \If{$N_k=\{1\}$} 
            \State  $z_k \gets x_1,\ z_{3-k}\gets x_3$ 
            \State \Return $(z_1,z_2)$ \label{line:m4-left-endpoint}
        \EndIf
        \If{$N_k=\{4\}$} 
            \State  $z_k \gets x_4,\ z_{3-k} \gets x_2$
            \State \Return $(z_1,z_2)$ \label{line:m4-right-endpoint}
        \EndIf
    \EndFor
    \State $O_L\gets(x_2,x_3)$, $O_R\gets(x_3,x_2)$
    \If{$\SC(O_L\mid A)\le \SC(O_R\mid A)$} \State \Return $O_L$\label{line:m4-middle-left}
    \Else \State \Return $O_R$\label{line:m4-middle-right}
    \EndIf
\Else\label{line:m4-large-gap}
    \State \Return $\EPM{}(A)$\label{line:m4-large-return}
\EndIf
\end{algorithmic}
\end{algorithm}

\begin{theorem}
\label{thm:four-small-sp}
Mechanism $M_{4}$ is strategyproof.
\end{theorem}
\begin{proof}
Observe that the branch depends only on the public locations and never on the reported approvals; hence, it creates no strategic opportunity by itself. Furthermore, by Corollary~\ref{cor:epm-sp}, \EPM{} is strategyproof, so it suffices to consider $M_{\leq 4}$. 

First observe that endpoint--endpoint deviations are harmless. If both profiles remain in the same endpoint branch, say $N_k=\{1\}$ or $N_k=\{4\}$, then the outcome is unchanged. A unilateral deviation cannot move directly from a left endpoint branch to a right endpoint branch, since this would require turning some singleton support $\{1\}$ into some singleton support $\{4\}$, or vice versa, which cannot be done by changing one report. 
It therefore suffices to rule out profitable deviations across two kinds of profile pairs: 
(1) profiles where both facilities are placed at the two middle agents, and 
(2) profiles such that in one of them a facility is placed at the leftmost agent (say $N_k=\{1\}$; the case $N_k=\{4\}$ is symmetric), while in the other profile both facilities are placed at the middle agents. 

\medskip
\noindent 
{\bf Middle--middle profiles.}
In this case, the mechanism selects the smaller-cost outcome between $O_L=(x_2,x_3)$ and $O_R=(x_3,x_2)$, with ties broken in favor of $O_L$. Consider the contribution of a single agent $i$ to $\SC(O_L\mid A)-\SC(O_R\mid A)$.
If $i$ approves both facilities, she pays $d(x_i,x_2)+d(x_i,x_3)$ under either outcome, so she is indifferent and contributes $0$. 
If $i$ approves only $\Fone$, she contributes $d(x_i,x_2)-d(x_i,x_3)$, whereas, if she approves only $\Ftwo$, she contributes its negation. 
A negative contribution favors $O_L$, while a positive contribution favors $O_R$. Thus, for a singleton agent, the truthful report already contributes towards the outcome among $O_L,O_R$ that she prefers: changing to the other singleton type reverses this contribution, while changing to $\Both$ removes it. Hence, no middle--middle deviation is profitable.

\medskip
\noindent 
{\bf Endpoint--middle profiles.}
By symmetry, it suffices to consider the left endpoint branch for facility $\Fone$. Thus let $A$ be a profile with $N_1=\{1\}$, so that
$M_{\le4}(A)=(x_1,x_3)$. Then agents $2,3,4$ all approve only $\Ftwo$. Let $A'$ be a profile obtained from $A$ by one deviation that is in the middle branch. We check the two strategyproofness inequalities for the pair $\{A,A'\}$. Throughout, any middle outcome uses the facility-node set $\{x_2,x_3\}$, while the endpoint outcome uses $\{x_1,x_3\}$.

\begin{itemize}
\item \emph{The truthful profile is $A$.}
If agent $1$ approves only $\Fone$, her cost in the endpoint outcome $(x_1,x_3)$ is already $0$. If agent $1$ approves both facilities, the endpoint outcome $(x_1,x_3)$ gives her cost $d(x_1,x_3)$, whereas either middle outcome using $\{x_2,x_3\}$ gives her cost $d(x_1,x_2)+d(x_1,x_3)$, which is larger. Agents $3$ and $4$ approve only $\Ftwo$, which is placed at $x_3$ in the endpoint outcome; since they are located weakly to the right of $x_3$, no middle outcome places $\Ftwo$ closer to them.

It remains to consider agent $2$, who also approves only $\Ftwo$. A deviation by agent $2$ that leaves the endpoint branch must add $\Fone$, so she reports either $\Both$ or $\One$. In the resulting middle-branch profile, agents $3$ and $4$, who approve only $\Ftwo$, each contribute $-g$ to $\SC(O_L)-\SC(O_R)$. Agent $1$ contributes either $-g$ or $0$, and agent $2$ contributes either $0$ if she reports $\Both$, or $-g$ if she reports $\One$. Hence $\SC(O_L)-\SC(O_R)\le -2g<0$, and the mechanism selects $O_L=(x_2,x_3)$. Since $F_2$ is again placed at $x_3$, agent
$2$'s cost remains unchanged.

\item \emph{The truthful profile is $A'$.}
The branch $N_1=\{1\}$ can be created only if the deviator either adds $\Fone$ (agent $1$) or drops $\Fone$ (one of agents $2,3,4$).
Let $i$ be the deviator agent. We consider all possible cases. 
\begin{itemize}
    \item $i$ is agent $1$. Then her true preferences do not contain $\Fone$, so it is $\Two$. In the endpoint outcome, $\Ftwo$ is at $x_3$,
which is no closer to $x_1$ than its middle-branch location (which is either $x_2$ or $x_3$).

    \item $i$ is agent $3$ or agent $4$. 
    We have that $x_i \geq x_3 \geq x_2 \geq x_1$, and deviating from $A'$ to $A$ changes the outcomes from $\{x_2,x_3\}$ to $(x_1,x_3)$. 
    If $i$ truly approves only $\Fone$, then the cost of $i$ cannot decrease since $d(x_i,x_3) \leq d(x_i,x_2) \leq d(x_i,x_1)$.
    Similarly, if $i$ truly approves both facilities, then $d(x_i,x_2)+d(x_i,x_3) \le d(x_i,x_1)+d(x_i,x_3)$.
    In any case, $i$ cannot gain. 

    \item $i$ is agent $2$. 
    Then she approves $\Fone$, and $N_1 = \{1,2\}$; otherwise dropping agent $2$'s approval of $\Fone$ would not yield $N_1=\{1\}$. 
    Hence, in $A'$, agents $3$ and $4$ contribute $-g$ each to $\SC(O_L)-\SC(O_R)$, agent $1$ contributes either $-g$ or $0$, and agent $2$ contributes either $-g$ if she reports $\One$, or $0$ if she reports $\Both$. Since the difference is negative, $O_L=(x_2,x_3)$ is selected.
    So, the contribution of $\Fone$ to agent $2$'s cost in the truthful profile $A'$ is $0$. In addition, $F_2$ is placed at $x_3$ for both $A$ and $A'$. Therefore, agent $2$ has no incentive to deviate. 
\end{itemize}
\end{itemize}
The proof is now complete. 
\end{proof}

We finally show the upper bound of $4/3$ on the approximation ratio. 

\begin{theorem}
    \label{lem:four-ratio}
The approximation ratio of $M_{4}$ is at most $4/3$.
\end{theorem}

\begin{proof}
We first analyze the approximation ratio when $M_{\leq 4}$ is executed. 
The mechanism is optimal for instances where one facility is placed at either the leftmost or the rightmost agent: 
if $N_k=\{1\}$, then $F_k$ is placed optimally at $x_1$, and the other facility is approved by agents $2,3,4$ (median $x_3$) or by
all four agents (median interval $[x_2,x_3]$), and thus its location $x_3$ is optimal; the case $N_k=\{4\}$ is symmetric.

Now, consider the case where the facilities are placed at $x_2$ and $x_3$. For every non-empty approver set $S \neq \{1\},\{4\}$, at least one of $x_2,x_3$ is median-optimal: 
$x_2$ is a median for $\{2\},\{1,2\},\{1,2,3\},\{1,2,4\}$, 
$x_3$ is a median for $\{3\},\{3,4\},\{1,3,4\},\{2,3,4\}$, 
and the median interval contains both $x_2$ and $x_3$ for every other $S \neq \{1\},\{4\}$. 
Hence, unless the same middle node is the unique median for both facilities, one of $O_L,O_R$ is socially optimal, and thus $M_{\le4}$ is optimal.

If $x_2$ is the unique median for the approver sets of both facilities, then $N_k = \{1,2,3\}$ and $N_{3-k} = \{1,2,4\}$. 
The combined one-facility cost at $x_2$ is $2\cdot d(x_1, x_2) + 2\cdot d(x_2, x_3) + d(x_3, x_4) \geq 2g+3$, and the requirement for distinct locations forces one facility to be placed one step away from the common unique median $x_2$, leading to \( \OPT\geq 2g+4\). The mechanism instead moves one of the facilities to \(x_3=x_2+g\), leading to $\SC(M_{\le4})=\OPT+g-1$, and the approximation ratio is at most  $1+\frac{g-1}{2g+4}\le 5/4$ as $g\leq 4$. The case where $x_3$ is unique median for the approver sets of both facilities is symmetric. Therefore, $M_{\le4}$ is $5/4$-approximate for $g\le4$.

It remains to analyze the branch in which \(M_4\) runs \(\EPM{}\), which happens only when \(g\ge5\). If neither facility suffers a parity loss, then $\EPM{}$ attains the optimum cost. It remains to consider profiles in which at least one facility suffers a parity loss. Such a loss can occur only for an odd-cardinality approver set, that is, a singleton or a triple. If a facility has triple support, its one-facility optimal cost is $d(x_1, x_3) \geq g+1$, or $ d(x_1, x_4) \geq g+2$, or $d(x_2, x_4)\geq g+1$. Since $g \geq 5$, all of these are at least $6$. If instead a facility has singleton support $N_k=\{i\}$, then the remaining three agents must all approve $F_{3-k}$, which therefore has a one-facility optimal cost of at least $g+1\ge6$. In any case, $\OPT\geq 6$ and  Lemma~\ref{lem:additive} implies $\SC(\EPM{})\le\OPT+2\le\frac43\cdot\OPT$.
\end{proof}

\section{General Graphs}
\label{sec:general}
We now switch to general connected graphs, for which we show that the best possible approximation ratio of deterministic strategyproof mechanisms is between $3/2$ and $2$. We first establish the upper bound (Section~\ref{sec:general:upper}), which is based on a mechanism that generalizes the {\sc Parity-Median} mechanism that we considered for the case of the line graph in the previous section. Our lower bound of $3/2$ is shown for the claw graph $K_{1,3}$ and shows a separation between the line graph and more complex graphs (Section~\ref{sec:general:lower}). 

\subsection{Upper Bound}
\label{sec:general:upper}
Here, we show the following upper bound for unrestricted graphs. 

\begin{theorem}
\label{thm:graph-upper}
For every connected graph $G$, there is a deterministic strategyproof mechanism with approximation ratio at most $2$.
\end{theorem}

\noindent 
The mechanism is defined differently for $n\le3$ (Section~\ref{sec:graph-small}) and for $n\ge4$ (Section~\ref{sec:graph-large}). 

\subsubsection{Instances with $n\le3$ Agents}
\label{sec:graph-small}

For instances with $n\leq 3$ agents, we consider the \SL{} mechanism, which returns an outcome $\bw = (w_1,w_2)$ whose cost vector $\bigl(\cost_1(\bw \mid A_1),\cost_2(\bw\mid A_2),\cost_3(\bw\mid A_3)\bigr)$ is lexicographically minimum according to the reported approval profile $A = (A_1,A_2,A_3)$. 

\begin{theorem}
\label{lem:serial-lex}
For any connected graph and $n\le3$, the {\normalfont \SL{}} mechanism is strategyproof and its approximation ratio is at most $2$.
\end{theorem}

\begin{proof}
To show strategyproofness, fix any agent $i \in \{1,2,3\}$. Given the approval preferences reported by the agents $1,\ldots,i-1$ preceding $i$ in the priority order, the mechanism computes a set $X_i$ of outcomes that are lexicographically optimal for these agents. Crucially, $X_i$ is independent of agent $i$'s reported preferences, the mechanism selects the  outcomes in $X_i$ minimizing agent $i$'s cost, and later agents only break ties among outcomes with the same reported cost for agent $i$. Hence, agent $i$ cannot strictly improve her true cost by misreporting. 

We next show the upper bound of $2$ on the approximation ratio. We first argue that when there are $n\leq 2$ agents, the mechanism outputs an optimal outcome. This is clearly true when $n=1$ since the mechanism minimizes the cost of this single agent by definition. When there are $n=2$ agents, observe that the two facilities are placed at nodes in the median-interval of the agents that approve them over the shortest path connecting $x_1$ and $x_2$. In particular, the possible locations of the two facilities are $x_1$, a direct neighbor node of $x_1$ (along the shortest path), or $x_2$. Exactly which facilities are placed at these nodes depends on the preferences of the agents, but it is always the case that the choice is a median node. 

It remains to consider instances with $n=3$ agents. 
Let $\bw =(w_1,w_2)$ be the outcome of the mechanism, and denote by $\bo=(o_1,o_2)$ an optimal outcome.
Since $\bw$ is optimal for the first two agents, we have that 
\begin{align}
\label{eq:graph-serial-12}
  \cost_1(\bw)+\cost_2(\bw)\le \cost_1(\bo)+\cost_2(\bo).
\end{align}
We consider a few cases depending on the preferences of the agents.

\medskip
\noindent 
{\bf Case $|A_1|=1$}. 
Without loss of generality, suppose that $A_1=\{1\}$, and thus $w_1=x_1$. 
If agent $3$ approves $F_1$, then, by the triangle inequality, we have
\begin{align*}
    d(x_3,w_1) = d(x_3,x_1) \leq d(x_3,o_1) + d(x_1,o_1) = d(x_3,o_1) + \cost_1(\bo). 
\end{align*}
The location of $F_2$ depends on the preference of agent $2$. If agent $2$ approves $F_2$, then $w_2 = x_2$, and thus, if agent $3$ approves $F_2$, by the triangle inequality again, we have
\begin{align*}
d(x_3,w_2)=d(x_3,x_2) \le d(x_3,o_2)+d(x_2,o_2) \leq d(x_3,o_2)+ \cost_2(\bo). 
\end{align*}
On the other hand, if agent $2$ does not approve $F_2$, then $A_2=\{1\}$, which implies that $F_2$ is irrelevant to agents $1$ and
$2$. If agent $3$ approves $F_2$, then, since $x_3\ne x_1$, $w_2 = x_3$, and therefore 
$$d(x_3,w_2) = 0.$$
Putting everything together, in any case, we have that
\begin{align*}
\cost_3(\bw) \le \cost_3(\bo)+ \cost_1(\bo)+\cost_2(\bo), 
\end{align*}
and thus, using also \eqref{eq:graph-serial-12}, we have
\begin{align*}
    \SC(\bw) 
    &= \cost_1(\bw)+\cost_2(\bw) + \cost_3(\bw) \\
    &\leq 2\cdot \cost_1(\bo)+2\cdot \cost_2(\bo) + \cost_3(\bo) \leq 2 \cdot \OPT.
\end{align*}

\medskip
\noindent 
{\bf Case $|A_1|=2$}. 
Then, one facility is placed at $x_1$ and the other at a direct neighbor of $x_1$. Without loss of generality, let $w_1=x_1$ and $d(w_1,w_2)=1$. 
Hence, using the triangle inequality, we have
\begin{align*}
    \cost_3(\bw) 
    &= \mathbf{1}_{\{1 \in A_3\}} \cdot d(x_3,w_1) +  \mathbf{1}_{\{2 \in A_3\}} \cdot d(x_3,w_2) \\
    &\leq  \mathbf{1}_{\{1 \in A_3\}} \cdot \bigg( d(x_3,o_1) + d(w_1,o_1) \bigg) +  \mathbf{1}_{\{2 \in A_3\}} \cdot \bigg( d(x_3,o_2) + d(w_1,w_2) + d(w_1,o_2) \bigg) \\
    &\leq \cost_3(\bo) + \cost_1(\bo) + 1.
\end{align*}
Now consider the following two cases:
\begin{itemize}
\item 
$\cost_2(\bo) + \cost_3(\bo) \geq 1$.
Then,
\begin{align*}
    \cost_3(\bw) \leq \cost_1(\bo) + \cost_2(\bo) + 2 \cdot \cost_3(\bo).
\end{align*}
Therefore, using also \eqref{eq:graph-serial-12}, we have
\begin{align*}
    \SC(\bw) 
    &= \cost_1(\bw) + \cost_2(\bw) + \cost_3(\bw) \\
    &\leq 2\cdot \cost_1(\bo)+2\cdot \cost_2(\bo) + 2\cdot \cost_3(\bo) = 2 \cdot \OPT.
\end{align*}

\item 
$\cost_2(\bo) + \cost_3(\bo) = 0$.
Then it has to be the case that agent $2$ and agent $3$ approve different facilities, and the optimal outcome places their approved facilities at their nodes. Since $|A_1|=2$, we therefore have that $\OPT = d(x_1,x_2) + d(x_1,x_3)$. 
In contrast, to minimize the cost of agent $2$, the mechanism places the facility that agent $3$ approves at $x_1$ and the other facility (that agent $2$ approves) at a neighbor of $x_1$ along the shortest path between $x_1$ and $x_2$. Consequently, 
\begin{align*}
    \SC(\bw) = 1 + d(x_1,x_2)-1 + d(x_1,x_3) = \OPT.
\end{align*}
\end{itemize}
In any case, the mechanism is at most $2$-approximate and the proof is complete. 
\end{proof}

\subsubsection{Instances with $n\ge4$ Agents}
\label{sec:graph-large}
We now switch to the case $n \geq 4$, and present mechanism \OICM{} that is based on the same design principle as {\sc Parity-Median} from Section~\ref{sec:line}.

The mechanism first partitions the set of nodes $V$ into two disjoint sets $L$ and $R$.
Let $X=\{x_1,\ldots,x_n\}$ denote the set of occupied nodes, and identify each agent with her occupied node, so that the support sets $N_1$ and $N_2$ may be viewed as subsets of $X$. The mechanism selects a maximal independent set $I$ of $G[X]$ according to an arbitrary fixed deterministic rule. It then extends $I$, again using a fixed deterministic rule by adding only empty nodes, to a maximal independent set $L$ of $G$. Finally, it sets $R=V\setminus L$.

Once $L$ and $R$ have been determined, the mechanism permanently restricts $F_1$ to $L$ and $F_2$ to $R$. Each facility with non-empty support is placed at a restricted median of its support within its assigned domain, using a fixed tie-breaking rule. A facility with empty support is placed at a fixed default node in its assigned domain. See Algorithm~\ref{alg:occupied-independent-cut-median}.

\begin{algorithm}[H]
\caption{\OICM{}}
\label{alg:occupied-independent-cut-median}
\begin{algorithmic}[1]
\Require Reported approval profile $(A_1,\ldots,A_n)$
\Ensure A feasible outcome $(z_1,z_2)$
\State Construct a maximal independent set $I$ of $G[X]$
\State Extend $I$ to a maximal independent set $L$ of $G$
\State $R \gets V \setminus L$ 
\State $D_1\gets L$ and $D_2\gets R$
\For{$k\in\{1,2\}$}
    \State $N_k\gets\{i:k\in A_i\}$
    \If{$N_k=\varnothing$}
        \State $z_k\gets q_k$, for a fixed default node $q_k\in D_k$
    \Else
        \State $z_k\gets\arg\min_{v\in D_k}\sum_{i\in N_k}\dist(x_i,v)$, with fixed tie-breaking
    \EndIf
\EndFor
\State \Return $(z_1,z_2)$
\end{algorithmic}
\end{algorithm}

The first useful property of the mechanism is that the two domains are non-empty (implying that the outcome is feasible) and both form dominating sets. 

\begin{lemma}
\label{lem:domains-dominating}
The sets $L$ and $R$ defined by {\normalfont \OICM{}} are non-empty dominating sets.
\end{lemma}

\begin{proof}
The set $L$ is non-empty because the independent set $I$ of $G[X]$ is non-empty. In addition, since $G$ is connected and has at least two nodes, no independent set can equal all of $V$, and thus $R = V\setminus L$ is non-empty as well.
The maximality of $L$ implies that every node of $R$ has a neighbor in $L$. Conversely, every node of $L$ has a neighbor since $G$ is connected, and this neighbor must be in $R$ since $L$ is an independent set. 
\end{proof}

The second property of the mechanism is restricted-median monotonicity, which underlies its strategyproofness, as well as the strategyproofness of the \EPM{} mechanism for the line.

\begin{lemma}
\label{lem:cut-mono}
Let $D\subseteq V$ be a fixed and non-empty set of nodes. For every non-empty set of agents $S$, let
\[
    g_D(S)\in\arg\min_{v\in D}\sum_{i\in S}\dist(x_i,v)
\]
with fixed tie-breaking. Then, for every $i\in S$,
\[
    \dist(x_i,g_D(S))
    \le
    \dist\bigl(x_i,g_D(S\setminus\{i\})\bigr),
\]
where $g_D(\varnothing)$ denotes the fixed default node of $D$.
\end{lemma}

\begin{proof}
Suppose first that $S\setminus\{i\}\ne\varnothing$.
The optimality of $g_D(S)$ for $S$ implies that
\[
    \sum_{j\in S}\dist(x_j,g_D(S))\le\sum_{j\in S}\dist(x_j,g_D(S\setminus\{i\})).
\]
Similarly, the optimality of $g_D(S\setminus\{i\})$ for $S\setminus\{i\}$ implies that
\[
    \sum_{j\in S\setminus\{i\}}\dist(x_j,g_D(S\setminus\{i\}))
    \le
    \sum_{j\in S\setminus\{i\}}\dist(x_j,g_D(S)).
\]
By adding these two inequalities, we obtain $\dist(x_i,g_D(S))\le\dist(x_i,g_D(S\setminus\{i\}))$.
Finally, if $S=\{i\}$, then $g_D(S)$ is one of the closest nodes of $D$ to $x_i$, and thus no default node can be strictly closer.
\end{proof}

We can now show that the mechanism is strategyproof. 

\begin{lemma}
\label{lem:oicm-sp}
The {\normalfont \OICM{}} mechanism is strategyproof.
\end{lemma}

\begin{proof}
The domains $L$ and $R$ are computed based only on public information. If an agent withdraws approval for a facility that she truly approves, Lemma~\ref{lem:cut-mono} shows that this facility cannot move strictly closer to her. On the other hand, if she adds approval for a facility that she does not truly approve, only a term outside her true cost is affected. Since costs are additive over approved facilities, no deviation is profitable.
\end{proof}

Next, we focus on bounding the approximation ratio. To simplify our notation in the analysis, for any set of agents $S \subseteq N$ and any set of nodes $U \subseteq V$, let 
\[
    f_U(S)=\min_{v\in U}\sum_{i\in S}\dist(x_i,v)
\]
be the minimum total distance of the agents in $S$ from any node of $U$. For completeness, define $f_U(\varnothing)=0$. We will also write $f(S)$ as shorthand for $f_V(S)$. Clearly, when $S \in \{N_1,N_2\}$, $f_U(S)$ gives us the total contribution of the agents in $S$ to either the social cost of the mechanism when $U \in \{L,R\}$, or to the optimal social cost when $U=V$. Note that, due to the constraint that facilities occupy distinct nodes, we have $\OPT\geq f(N_1)+f(N_2)$. In addition, since agents occupy different nodes of the graph, every candidate node has distance zero from at most one agent of $S$, implying that $f(S) \geq |S|-1$; this will come in handy in the proofs of the following structural lemmas. 

We next establish three bounds comparing the cost of serving a set of agents from one of the restricted domains with its optimal unrestricted one-facility cost. We begin with a general estimate that applies to any dominating domain (and can thus be applied to both $L$ and $R$).

\begin{lemma}
\label{lem:dominating-rounding}
Let $D\subseteq V$ be a dominating set. For every set of agents $S$,
\[
    f_D(S)\le 2\cdot f(S)+1.
\]
\end{lemma}

\begin{proof}
The statement is immediate for $S=\varnothing$ since, in this case, $f(S)=0$ and $f_D(S) \leq 1$. 
So, suppose that $S\ne\varnothing$, and let $v$ be an unrestricted median of $S$. If $v\in D$, then $f_D(S)=f(S)$. Otherwise, if $v \not\in D$, the fact that $D$ is a dominating set implies that $v$ has a neighbor $y\in D$.
Therefore, since $f_D(S)$ is the minimum total distance of the agents in $S$ from any node in $D$, 
\[
    f_D(S)
    \le \sum_{i\in S}d(x_i,y)
    \le f(S)+|S|.
\]
Since $f(S)\ge |S|-1$, we obtain
\[
    f_D(S)\le 2\cdot f(S)+1,
\]
and the proof is complete.
\end{proof}

The next two lemmas provide the improved  bound $f_D(S)\leq 2\cdot f(S)-1$ under conditions that, as we show in Theorem~\ref{thm:oicm-two}, cannot fail for both domains simultaneously; hence, whenever Lemma~\ref{lem:dominating-rounding} is tight for one of 
$L$ and $R$, the other saves the unit that is lost. For the domain $R$, the fact that its complement $L$ is an independent set allows us to obtain this sharper bound. In particular, when there are at least two supporters, we can obtain a bound strictly better than $2$.

\begin{lemma}
\label{lem:independent-complement}
For the set of nodes $R$ as computed by {\normalfont \OICM{}}, and any set of agents $S$,
\begin{align*}
    f_R(S) \le 
    \begin{cases}
        1, & \text{if } |S| \leq 1; \\
        2\cdot f(S)-1,  & \text{if } |S| \geq 2. \\
    \end{cases}
\end{align*}
\end{lemma}

\begin{proof}
The case $|S| \leq 1$ is immediate: $f_R(\varnothing)=0$, while a single agent in $S$ either lies in $R$ (and thus $f_R(S)=0$) or has a neighbor in $R$ (and thus $f_R(S) \leq 1$).

Now consider the case $|S|\ge2$. Let $v$ be an unrestricted median of $S$ so that $\sum_{i \in S} d(i,v) = f(S)$. 
If $v\in R$, then $f_R(S)=f(S)$. Since agents occupy distinct nodes and $|S|\geq 2$, we have that $f(S) \geq 1$, and thus $f_R(S) \le2\cdot f(S)-1$. So, we can assume that $v\in L$. 
Choose an agent $j \in S$ such that $j \neq v$. 
Let $y$ be the first node on a shortest path from $v$ to $j$. 
Since $L$ is an independent set, it must be that $y\in R$. 
Hence, since $f_R(S)$ is the minimum total distance of all agents in $S$ from any node in $R$,
\begin{align*}
    f_R(S) \leq \sum_{i \in S} d(x_i,y).
\end{align*}
Moving from $v$ to $y$ decreases the distance of $j$ by one and increases the distance of every other agent in $S$ by at most one. Hence, using also the fact that $f(S) \geq |S|-1$, we have
\[
    \sum_{i \in S} d(x_i,y) \le f(S)-1+(|S|-1)=f(S)+|S|-2\le2\cdot f(S)-1
\]
and the proof is complete.
\end{proof}

The remaining difficult case arises when the facility restricted to
$R$ is approved by at most one agent. Then, the facility restricted to $L$ is
either approved by all agents, or by all but one of them. The construction of
$I$ inside the occupied subgraph provides the additional structure
needed to obtain the following stronger bound.

\begin{lemma}
\label{lem:almost-complete-support}
Let $I$, $X$ and $L$ be as defined by {\normalfont \OICM{}}, with $I\subseteq L$. Suppose that
\[
    S=X
    \qquad\text{or}\qquad
    S=X\setminus\{x\}
\]
for some $x\in X$. If $|S|\ge3$, then
\[
    f_L(S)\le 2\cdot f(S)-1.
\]
\end{lemma}

\begin{proof}
Let $v$ be an unrestricted median of $S$, and recall that $f(S)\ge|S|-1\ge2$.
If $v\in L$, then we immediately have that $f_L(S)=f(S)\le 2\cdot f(S)-1$.
Hence, assume that $v\notin L$; since $I\subseteq L$, we also have $v\notin I$. As $L$ is a
dominating set, $v$ has a neighbor $y \in L$, and therefore, since $f_L(S)$ is the minimum total distance of the agents in $S$ from any node in $L$, 
\[
    f_L(S) \le \sum_{i \in S} d(i,y) \le \sum_{i \in S} \bigg( d(i,v) + 1 \bigg) = f(S)+|S|.
\]
Consequently, if $f(S)\ge|S|+1$, then $f_L(S)\le 2\cdot f(S)-1$. For the remainder of the proof, we may therefore assume that $f(S) \le |S|$. We consider the following exhaustive cases, and for each of them show the desired bound. 

\medskip
\noindent
{\bf Case 1: $I\cap S=\varnothing$.}
Since $I$ is non-empty, $I\subseteq X$, and $I\cap S=\varnothing$, it cannot be
that $S=X$; hence $S=X\setminus\{x\}$ for some $x\in X$. Moreover,
$I\subseteq X\setminus S=\{x\}$, and thus $I=\{x\}$. Since $I$ is a maximal independent set of
$G[X]$, node $x$ must be adjacent to every agent in $S$, and therefore, as
$x\in I\subseteq L$,
\[
    f_L(S)\le\sum_{i\in S}d(x_i,x)=|S|\le2|S|-3\le2\cdot f(S)-1,
\]
where the second to last inequality follows by the fact that $|S|\ge3$, and the last inequality follows since $f(S)\ge|S|-1$.

\medskip
\noindent
{\bf Case 2: $I\cap S \ne \varnothing$ and $d(a,v)=1$ for some $a \in I \cap S$.}
Consider placing the facility at $a$. Relative to placing it at $v$, the distance of the agent at $a$ decreases by one, while the distance of every other agent in $S$ increases by at most one. Hence, since $f(S) \ge |S|-1$,
\[
    f_L(S)\le f(S)-1+(|S|-1) \le2\cdot f(S)-1.
\]

\medskip
\noindent
{\bf Case 3: $I\cap S\ne\varnothing$ and $d(a,v)\ge2$ for every $a\in I\cap S$.}
We first make the following structural claim. 

\medskip
\noindent
{\bf Claim A.}
{\it At most one agent is at distance at least $2$ from $v$. Moreover, if such an agent $w$ exists, then $v\in S$, $f(S)=|S|$, $d(w,v)=2$, and every node of $S\setminus\{v,w\}$ is adjacent to $v$.}

\smallskip
\noindent
{\it Proof of Claim A.}
Let $T=\{u\in S:d(u,v)\ge2\}$ be all the agents in $S$ that are at distance at least $2$ from $v$. 
Since agents occupy distinct nodes, at most one agent is located at $v$. Each agent $i \not\in T \cup \{v\}$ contributes at least one to $f(S)$, while every agent in $T$ contributes at least $2$. Therefore,
\[
    f(S)\ge|S|-|T|-\mathbf{1}_{\{v\in S\}}+2|T|=|S|+|T|-\mathbf{1}_{\{v\in S\}}.
\]
Since $f(S)\le|S|$, it follows that $|T|\le\mathbf{1}_{\{v\in S\}}\le1$. If $T\ne\varnothing$, then necessarily $|T|=1$ and $v\in S$, and equality must hold throughout. Hence, $f(S)=|S|$ and $d(w,v)=2$ for the unique node $w\in T$, while every agent in $S\setminus\{v,w\}$ is at distance exactly $1$ from $v$.
\hfill$\lhd$

\medskip
Since $I\cap S\ne\varnothing$ and every node of $I\cap S$ is at distance at least $2$ from $v$, Claim~A implies that $|I\cap S|=1$. In particular, $I \cap S = \{w\}$, $v\in S$, $f(S)=|S|$, $d(w,v)=2$, and every agent in $S\setminus\{v,w\}$ is adjacent to $v$.

\medskip
\noindent
{\bf Claim B.}
{\it There exists a node $x\in X$ such that $S=X\setminus\{x\}$, $x\in I$, and $d(x,v)=1$. Consequently $I=\{w,x\}$, and every node of $S\setminus\{v,w\}$ is adjacent to at least one of $w$ and $x$.}

\smallskip
\noindent
{\it Proof of Claim B.}
Since $v\in S\subseteq X$ and $v\notin I$, the maximality of $I$ in $G[X]$ implies that $v$ has a neighbor in $I$. Since $I\cap S=\{w\}$ and $d(w,v)=2$, this neighbor cannot be $w$. Therefore, $S=X\setminus\{x\}$ for some $x\in X$, and the neighbor of $v$ in $I$ must be this node $x$, so that $x\in I$ and $d(x,v)=1$. Since $I\cap S=\{w\}$ and $X\setminus S=\{x\}$, we obtain $I=\{w,x\}$. Finally, due to the fact that $S\setminus\{v,w\} = X\setminus\{x,v,w\} \subseteq X\setminus I$ and the maximality of $I=\{w,x\}$ in $G[X]$, every node of $S\setminus\{v,w\}$ is adjacent to at least one of $w$ and $x$.
\hfill$\lhd$

\medskip
By Claim~B, every node of $S\setminus\{v,w\}$ is adjacent to at least one of $w$ and $x$.
Suppose first that some node $y \in S\setminus\{v,w\}$ is adjacent to $x$, and consider placing the facility at $x$. We have $d(v,x)=1$, $d(y,x)=1$, and, by the triangle inequality, $d(w,x)\le d(w,v)+d(v,x)=3$. Each of the remaining agents $i \in S\setminus\{v,w,y\}$ is adjacent to $v$, and thus $d(i,x) \leq d(i,v) + d(v,x) = 2$. 
Putting all of these together, and since $f(S) = |S|$, 
\[
    f_L(S)\le \sum_{i \in S} d(x_i,x) \le 1+1+3+2(|S|-3)=2|S|-1=2\cdot f(S)-1.
\]
Finally, suppose that no node of $S\setminus\{v,w\}$ is adjacent to $x$. By Claim~B, every such node must then be adjacent to $w$, and therefore
\[
    f_L(S)\le \sum_{i \in S} d(x_i,w) \le d(v,w)+\sum_{i\in S\setminus\{v,w\}}d(i,w)
    \le 2+(|S|-2)=|S| < 2\cdot f(S)-1,
\]
where the last inequality follows by the fact that $f(S) = |S| > 1$. 
This completes the proof.
\end{proof}

We can now combine the three bounds proved in the lemmas above.
When there are at least two agents approving $F_2$, the additive loss for the dominating domain $L$ (Lemma~\ref{lem:dominating-rounding}) is canceled by the sharper estimate for $R$ (Lemma~\ref{lem:independent-complement}). On the other hand, when there is at most one agent approving $F_2$, the bound in Lemma~\ref{lem:almost-complete-support} supplies the corresponding improvement directly for $L$. These lead to the desired bound of $2$ on the approximation ratio of \OICM{}.

\begin{theorem}
\label{thm:oicm-two}
For every connected graph and $n\ge4$, the approximation ratio of the {\normalfont \OICM{}} mechanism is at most $2$.
\end{theorem}

\begin{proof}
Let $\bw=(w_1,w_2)$ be the outcome of \OICM{}. Since $\Fone$ is restricted to $L$ and $\Ftwo$ is restricted to $R$,
\[
    \SC(\bw) = f_L(N_1)+f_R(N_2).
\]
By the separable lower bound established above,
\[
    \OPT\ge f(N_1)+f(N_2).
\]
We distinguish two cases.

\medskip
\noindent
{\bf Case $|N_2|\ge2$.} 
By Lemma~\ref{lem:dominating-rounding} for $L$ and Lemma~\ref{lem:independent-complement} for $R$ with $|N_2|\geq 2$, we obtain 
\begin{align*}
    \SC(\bw)
    &=f_L(N_1)+f_R(N_2)\\
    &\le2f(N_1)+1+2f(N_2)-1\\
    &=2\bigl(f(N_1)+f(N_2)\bigr)\\
    &\le 2\cdot \OPT.
\end{align*}

\medskip
\noindent 
{\bf Case $|N_2|\le1$.}
Since every agent approves at least one
facility, when $|N_2|=0$ we have $N_1=N$, while when $|N_2|=1$ and for the unique agent $i \in N_2$, we have either $N_1=N$ if $i$ approves both facilities or $N_1=N\setminus\{i\}$ when $i$ approves only $F_2$.

In other words, by identifying agents with their occupied nodes,
\[
    N_1=X
    \qquad\text{or}\qquad
    N_1=X\setminus\{x_i\}.
\]
Since $n\ge4$, we have $|N_1|\ge3$, and
Lemma~\ref{lem:almost-complete-support} gives us that
\[
    f_L(N_1)\le 2\cdot f(N_1)-1.
\]
Moreover, by Lemma~\ref{lem:independent-complement}, we have that $f_R(N_2) \le 1$.
Putting these together, we obtain
\begin{align*}
    \SC(\bw) \le2\cdot f(N_1)-1+1 \le 2\cdot \OPT.
\end{align*}
This completes the proof.
\end{proof}

\subsection{Lower Bound}
\label{sec:general:lower}
We now prove the following lower bound for graphs beyond the line. 

\begin{theorem}\label{thm:lower:3/2:general-graphs}
For connected graphs, the approximation ratio of any deterministic strategyproof mechanism is at least $3/2$.
\end{theorem}

We consider instances on the claw graph $K_{1,3}$, which consists of a central node $0$ and leaves $1,2,3$; we have that $\dist(0,j)=1$ for each leaf $j$ and $\dist(i,j)=2$ for two distinct leaves $i$ and $j$. There is one agent at each node. We will write $(t_0,t_1,t_2,t_3)$ to represent the approval profile of the four agents, where $t_v\in\{\One,\Two,\Both\}$. For convenience, we will denote by $B=\Both$ the doubleton approval preference.  

Assume towards a contradiction that there is a strategyproof mechanism $M$ with approximation ratio strictly below $3/2$. 
We will say that a feasible outcome $\mathbf{z}=(z_1,z_2)$ is \emph{admissible} for a profile $A$ if $\SC(\mathbf{z}\mid A) < (3/2)\cdot\OPT(A)$. Hence, since $M$ is strictly better than $3/2$-approximate, $M(A)$ must be admissible for every profile $A$.
We start with the all-doubleton profile
\[
P_0=(B,B,B,B).
\]
Up to permuting the leaves of the graph and swapping the names of the two facilities, there are two cases to consider: 
\begin{itemize}
\item {\bf (case LL)} both facilities are placed on leaves;
\item {\bf (case CL)} one facility is placed at the center and one on a leaf. 
\end{itemize}
Therefore, without loss of generality, it suffices to assume that $M(P_0)=(1,2)$ for {\bf (case LL)}, and $M(P_0)=(0,1)$ for {\bf (case CL)}.

We first record a set of \emph{anchor} profiles with a unique admissible outcome.

\begin{lemma}\label{lem:anchors}
Each of the following profiles has a unique admissible outcome: 
\renewcommand{\arraystretch}{1.2}
\begin{center}
\begin{tabular}{c|c}
{\normalfont profile} & {\normalfont admissible outcome} \\\hline
 $(\One,\One,\Two,\One)$  & $(0,2)$ \\
 $(\Two,\One,\Two,\Two)$  & $(1,0)$ \\
 $(\Two,\Two,\Two,\One)$  & $(3,0)$ \\
 $(B,\One,\One,B)$  & $(0,3)$ \\
 $(B,\One,B,\One)$ & $(0,2)$ \\
 $(B,\Two,\Two,B)$  & $(3,0)$ \\
 $(B,\Two,B,\Two)$ & $(2,0)$ \\
 $(B,B,\One,\One)$  & $(0,1)$ \\
 $(B,B,\Two,\Two)$  & $(1,0)$ 
\end{tabular}
\end{center}
\end{lemma}

\begin{proof}
For the first three singleton-only profiles, the optimal social cost is $2$, while every other outcome has social cost at least $3$.
For the remaining profiles, the optimal social cost is $4$, while every other outcome has social cost at least $6$.  
Since the approximation ratio is strictly below $3/2$, only the presented optimal outcomes are admissible.
\end{proof}

Our proof argues about the set of admissible outcomes of the  \emph{terminal} profile 
\[
Q=(B,\One,B,B).
\]
We begin by eliminating $Q$'s  admissible outcomes that are incompatible, due to strategyproofness, with the anchor profiles in Lemma~\ref{lem:anchors}.

\begin{lemma}\label{lem:Q}
For the terminal profile $Q=(B,\One,B,B)$, the only outcomes that may be chosen by $M$ are
$(0,1), (0,2), (0,3), (1,0).$
\end{lemma}

\begin{proof}
First, observe that $\OPT(Q)=6$. Since the social cost of outcomes $(2,1)$ and $(3,1)$ is $10 > 9=(3/2)\cdot\OPT(Q)$, the admissible outcomes for $Q$ are 
\[
(0,1),(0,2),(0,3),(1,0),(1,2),(1,3),(2,0),(2,3),(3,0),(3,2).
\]
We now reduce this set of admissible outcomes using the anchor profile $(B,\One,\One,B)$ whose unique admissible outcome is $(0,3)$, and the anchor profile $(B,\One,B,\One)$ whose unique admissible outcome is $(0,2)$. 
Profile $Q=(B,\One,B,B)$ differs from $(B,\One,\One,B)$ only at the leaf-$2$ agent $i$, and from $(B,\One,B,\One)$ only at the leaf-$3$ agent $j$. Consider the following cases: 
\begin{itemize}
\item 
$M(Q)=(1,3)$. 
If $Q$ is the truthful profile, $i$ has true preference $B$ and cost $4$. 
By misreporting her preference as $\One$, the profile changes to $(B,\One,\One,B)$ and the outcome becomes $(0,3)$, for which $i$ has cost $3$. Hence, $i$ has incentive to deviate, contradicting that $M$ is strategyproof.  

\item 
$M(Q)=(2,0)$ or $M(Q)=(2,3)$.
If $(B,\One,\One,B)$ is the truthful profile, $i$ has true preference $\One$ and cost $1$ for the outcome $(0,3)$.
By misreporting her preference as $B$, the profile changes to $Q$ with outcome $(2,0)$ or $(2,3)$, for both of which $i$ has cost $0$.
This again contradicts that $M$ is strategyproof.  

\item $M(Q)=(1,2)$. 
If $Q$ is the truthful profile, $j$ has true preference $B$ and cost $4$.
By misreporting her preference as $\One$, the profile changes to $(B,\One,B,\One)$ and the outcome becomes $(0,2)$, for which $j$ has cost $3$, a contradiction.

\item  
$M(Q) = (3,0)$ or $M(Q)=(3,2)$.
If  $(B,\One,B,\One)$ is the truthful profile, $j$ has true preference $\One$ and cost $1$ for the outcome $(0,2)$.  
By misreporting her preference as $B$, the profile changes to $Q$ with outcome $(3,0)$ or $(3,2)$, for both of which $j$ has cost $0$, again a contradiction.
\end{itemize}
Hence, only the admissible outcomes $(0,1),(0,2),(0,3),(1,0)$ remain.
\end{proof}

We are now ready to handle {\bf (case LL)}; recall that, in this case, we assume that $M(P_0) = (1,2)$.

\begin{lemma}\label{lem:caseA}
If $M(P_0)=(1,2)$, the approximation ratio of $M$ is at least $3/2$.
\end{lemma}

\begin{proof}
Profiles $P_0=(B,B,B,B)$ and $Q=(B,\One,B,B)$ differ only at the preference of leaf-$1$ agent $i$. 
By Lemma~\ref{lem:Q}, the only possible admissible outcomes for $Q$ are $(0,1),(0,2),(0,3),(1,0)$.
\begin{itemize}
\item $M(Q) \in \{ (0,1),(0,2),(0,3) \}$.
If $Q$ is the truthful profile, $i$ has true preference $\One$ and cost $1$ for any of these possible outcomes. 
By misreporting $B$, the profile changes to $P_0$ and the outcome becomes $(1,2)$, for which $i$ has cost $0$.
Therefore, $i$ has incentive to deviate, contradicting that $M$ is strategyproof. 

\item $M(Q) = (1,0)$. 
If $P_0$ is the truthful profile, $i$ has true preference $B$ and her cost is $2$ for the outcome $(1,2)$. 
By misreporting $\One$, the profile changes to $Q$ with outcome $(1,0)$, for which her cost is $1$. This again contradicts that $M$ is strategyproof.  
\end{itemize}
Overall, none of the possible admissible outcomes for $Q$ can be chosen by $M$, and thus $M$ must output a non-admissible outcome, leading to an  approximation ratio of at least $3/2$.
\end{proof}

We next switch to {\bf (case CL)} where $M(P_0) = (0,1)$. 
Contrary to {\bf (case LL)}, here the analysis is laborious and requires arguing about a chain of profiles, until we reach $Q$.

\begin{lemma}\label{lem:caseB}
If $M(P_0)=(0,1)$, the approximation ratio of $M$ is at least $3/2$.
\end{lemma}

\begin{proof} 
By Lemma~\ref{lem:Q}, the only possible admissible outcomes for $Q=(B,\One,B,B)$ are $(0,1)$, $(0,2)$, $(0,3)$, $(1,0)$. We will argue that two neighboring profiles of $Q$, $(B,\One,\Two,B)$ and $(B,\One,B,\Two)$, have a unique admissible outcome that is compatible with strategyproofness. In particular, we show the following properties; the proof of this lemma follows by arguing about the outcomes of several other profiles and can be found in the appendix. 

\begin{lemma}\label{lem:lower:3/2:general-graph:neighbor-profiles}
$M(B,\One,\Two,B)=(0,2)$ and $M(B,\One,B,\Two)=(0,3)$.
\end{lemma}

\noindent 
Using these, we can eliminate the four possible admissible outcomes for $Q$, as follows:
\begin{itemize} 
    \item $M(Q) \in \{(0,1), (0,3), (1,0)\}$. 
    If the truthful profile is $Q$, then the leaf-$2$ agent with true preference $B$ has cost $3$.  
    By misreporting $\Two$, the profile changes to $(B,\One,\Two,B)$ with outcome $(0,2)$, and the cost of the agent reduces to $1$. This contradicts that $M$ is strategyproof. 

    \item $M(Q) = (0,2)$. 
    If the truthful profile is $Q$, then the leaf-$3$ agent with true preference $B$ has cost $3$.
    By misreporting $\Two$, the profile changes to $(B,\One,B,\Two)$ with outcome $(0,3)$, and the cost of the agent reduces to $1$. This again contradicts that $M$ is strategyproof.
\end{itemize}
Given that no admissible outcome is possible for $Q$, the approximation ratio of $M$ when given $Q$ as input must be at least $3/2$. 
\end{proof}

By Lemma~\ref{lem:caseA}, which handles {\bf (case LL)}, and Lemma~\ref{lem:caseB}, which handles {\bf (case CL)}, we can finally conclude that the approximation ratio of $M$ is at least $3/2$, and the proof of Theorem~\ref{thm:lower:3/2:general-graphs} is now complete.

\section{Conclusion}\label{sec:conclusion}
In this paper, we settled the optimal deterministic approximation ratio for discrete heterogeneous two-facility location on the line: the mechanism $\Mfourthirds$ is $4/3$-approximate for every number of agents, matching the known lower bound. Its core is the \EPM{} mechanism, which restricts the facilities to opposite parity classes and places each at a restricted median of its supporters; local overrides handle the remaining small instances. Beyond the line, for instances with $n \geq 4$ agents, we introduced the \OICM{} mechanism, which similarly restricts the two facilities to different domains constructed based on a maximal independent set of the occupied subgraph. Together with \SL{} for instances with $n \leq 3$ agents, this leads to a deterministic strategyproof $2$-approximation on every connected graph. In the other direction, we showed that on the claw graph $K_{1,3}$ no deterministic strategyproof mechanism achieves a ratio below $3/2$: the bound of $4/3$ does not extend even to trees with a single branching node.

The main immediate open problem is to close the gap $[3/2,2]$ for general connected graphs. Further directions include randomized mechanisms on graphs, where the lower bound of $3/2$ no longer applies; $k\ge3$ facilities, for which our framework asks for a partition of the nodes into $k$ mutually dominating parts, whose existence is no longer guaranteed; objectives beyond the social cost; and the original domain of \citet{serafino-ventre}, in which agents may approve neither facility.

\paragraph{LLM Usage Disclosure.}
ChatGPT 5.5/5.6 and Claude Opus 4.8 were used during exploratory searches for some of the mechanisms and for the claw lower-bound construction. All stated mechanisms, examples, and proofs were independently checked, completed, and written by the authors.

\bibliographystyle{plainnat}
\bibliography{references}

@article{Kanellopoulos2026obnoxious,
  author       = {Panagiotis Kanellopoulos and
                  Alexandros A. Voudouris},
  title        = {Constrained Truthful Obnoxious Two-Facility Location with Optional
                  Preferences},
  journal      = {Algorithmica},
  volume       = {88},
  number       = {3},
  pages        = {42},
  year         = {2026}
}

@article{serafino-ventre,
  author       = {Paolo Serafino and
                  Carmine Ventre},
  title        = {Heterogeneous facility location without money},
  journal      = {Theoretical Computer Science},
  volume       = {636},
  pages        = {27--46},
  year         = {2016}
}

@article{Kanellopoulos2023sidma,
  author       = {Panagiotis Kanellopoulos and
                  Alexandros A. Voudouris and
                  Rongsen Zhang},
  title        = {On Discrete Truthful Heterogeneous Two-Facility Location},
  journal      = {{SIAM} Journal on Discrete Mathematics},
  volume       = {37},
  number       = {2},
  pages        = {779--799},
  year         = {2023}
}

@article{Deligkas2025agent,
  author       = {Argyrios Deligkas and
                  Mohammad Lotfi and
                  Alexandros A. Voudouris},
  title        = {Agent-constrained truthful facility location games},
  journal      = {Journal of Combinatorial Optimization},
  volume       = {49},
  number       = {2},
  pages        = {24},
  year         = {2025}
}

@inproceedings{fl-survey,
  author    = {Hau Chan and
               Aris Filos{-}Ratsikas and
               Bo Li and
               Minming Li and
               Chenhao Wang},
  title     = {Mechanism design for facility location problems: {A} survey},
  booktitle = {Proceedings of the Thirtieth International Joint Conference on Artificial
               Intelligence {(IJCAI)}},
  pages     = {4356--4365},
  year      = {2021}
}

@article{procaccia09approximate,
  author    = {Ariel D. Procaccia and
               Moshe Tennenholtz},
  title     = {Approximate Mechanism Design without Money},
  journal   = {{ACM} Transactions on Economics and Computation},
  volume    = {1},
  number    = {4},
  pages     = {18:1--18:26},
  year      = {2013},
}

@article{Xu2021minimum,
  author       = {Xinping Xu and
                  Bo Li and
                  Minming Li and
                  Lingjie Duan},
  title        = {Two-facility Location Games with Minimum Distance Requirement},
  journal      = {Journal of Artificial Intelligence Research},
  volume       = {70},
  pages        = {719--756},
  year         = {2021}
}

@article{kanellopoulos2025,
  author       = {Panagiotis Kanellopoulos and
                  Alexandros A. Voudouris and
                  Rongsen Zhang},
  title        = {Truthful Two-Facility Location with Candidate Locations},
  journal     = {Theoretical Computer Science},
  volume =      {1024},
  pages = {114913},
  year          = {2025}
}

@article{ZhaoLNF24,
  author       = {Qi Zhao and
                  Wenjing Liu and
                  Qingqin Nong and
                  Qizhi Fang},
  title        = {Constrained heterogeneous two-facility location games with sum-variant},
  journal      = {Journal of Combinatorial Optimization},
  volume       = {47},
  number       = {4},
  pages        = {65},
  year         = {2024}
}

@inproceedings{GaiLW22,
  author       = {Ling Gai and
                  Mengpei Liang and
                  Chenhao Wang},
  title        = {Obnoxious Facility Location Games with Candidate Locations},
  booktitle    = {Proceedings of the 16th International
                  Conference Algorithmic Aspects in Information and Management ({AAIM})},
  volume       = {13513},
  pages        = {96--105},
  year         = {2022},
}

@inproceedings{Tang2020candidate,
  author       = {Zhongzheng Tang and
                  Chenhao Wang and
                  Mengqi Zhang and
                  Yingchao Zhao},
  title        = {Mechanism Design for Facility Location Games with Candidate Locations},
  booktitle    = {Proceedings of the 14th International Conference on Combinatorial Optimization and Applications ({COCOA})},
  pages        = {440--452},
  year         = {2020}
}

@article{lotfi2024max,
  author       = {Mohammad Lotfi and
                  Alexandros A. Voudouris},
  title        = {On Truthful Constrained Heterogeneous Facility Location with Max-Variant
                  Cost},
  journal      = {Operations Research Letters},
  volume       = {52},
  pages         = {107060},
  year         = {2024},
}

@article{Zhao2023constrained,
  author       = {Qi Zhao and
                  Wenjing Liu and
                  Qingqin Nong and
                  Qizhi Fang},
  title        = {Constrained heterogeneous facility location games with max-variant
                  cost},
  journal      = {Journal of Combinatorial Optimization},
  volume       = {45},
  number       = {3},
  pages        = {90},
  year         = {2023}
}

@inproceedings{Sui2015constrained,
  author       = {Xin Sui and
                  Craig Boutilier},
  title        = {Approximately Strategy-proof Mechanisms for (Constrained) Facility Location},
  booktitle    = {Proceedings of the 2015 International Conference on Autonomous Agents
                  and Multiagent Systems ({AAMAS})},
  pages        = {605--613},
  year         = {2015},
}

@article{FangFLNV2025predictions,
  author       = {Jiazhu Fang and
                  Qizhi Fang and
                  Wenjing Liu and
                  Qingqin Nong and
                  Alexandros A. Voudouris},
  title        = {Mechanism Design with Predictions for Facility Location Games with
                  Candidate Locations},
  journal      = {Journal of Combinatorial Optimization},
  volume       = {49},
  number       = {5},
  pages        = {74},
  year         = {2025}
}

@article{moulin1980,
   author  = {Herv{\'e} Moulin},
   title   = {On strategy-proofness and single peakedness},
   journal = {Public Choice},
   volume  = {35},
   number  = {4},
   pages   = {437--455},
   year    = {1980}
 }

@article{schummer-vohra2002,
   author  = {James Schummer and Rakesh V. Vohra},
   title   = {Strategy-proof location on a network},
   journal = {Journal of Economic Theory},
   volume  = {104},
   number  = {2},
   pages   = {405--428},
   year    = {2002}
 }

@article{fotakis-tzamos-2014,
   author  = {Dimitris Fotakis and Christos Tzamos},
   title   = {On the power of deterministic mechanisms for facility location games},
   journal = {ACM Transactions on Economics and Computation},
   volume  = {2},
   number  = {4},
   pages   = {15:1--15:37},
   year    = {2014}
 }

@inproceedings{dokow-etal-2012,
   author    = {Elad Dokow and Michal Feldman and Reshef Meir and Ilan Nehama},
   title     = {Mechanism design on discrete lines and cycles},
   booktitle = {Proceedings of the 13th ACM Conference on Electronic Commerce (EC)},
   pages     = {423--440},
   year      = {2012}
 }

\appendix

\section*{Proof of Lemma~\ref{lem:lower:3/2:general-graph:neighbor-profiles}}
We will show the lemma for the profile $(B,\One,\Two,B)$ by arguing about the admissible outcomes of neighboring profiles due to strategyproofness. The other profile $(B,\One,B,\Two)$ can be handled similarly by relabeling the leaves. Recall the standing assumption $M(P_0)=(0,1)$ of Lemma~\ref{lem:caseB}. 

\begin{claim}\label{cl:B121}
$M(B,\One,\Two,\One)=(0,2)$.
\end{claim}

\begin{proof}
The admissible outcomes of $(B,\One,\Two,\One)$ are
$(0,2),(1,0),(1,2),(3,0),(3,2)$. Indeed, the optimal outcome is $(0,2)$ with a social cost of $3$, and any outcome not listed has social cost at least $5$. We will show that for each of these outcomes besides $(0,2)$, there is a profitable deviation, contradicting that $M$ is strategyproof. 
\begin{itemize}
\item $(1,2)$ and $(3,2)$. 
If the truthful profile is $(B,\One,\Two,\One)$, the central agent has true preference $B$ and cost $2$. 
By misreporting $\One$, she can change to the anchor profile $(\One,\One,\Two,\One)$ and the outcome becomes $(0,2)$ due to Lemma~\ref{lem:anchors}. Hence, she can decrease her cost to $1$.

\item $(1,0)$ and $(3,0)$.
If the truthful profile is $(B,\One,\Two,\One)$, the leaf-$2$ agent has true preference $\Two$ and cost $1$. 
By misreporting $B$, she can change to the anchor profile $(B,\One,B,\One)$ and the outcome becomes $(0,2)$ due to Lemma~\ref{lem:anchors}. Hence, she can decrease her cost to $0$.
\end{itemize}
The only remaining admissible outcome is $(0,2)$.
\end{proof}

\begin{claim}\label{cl:B122} 
$M(B,\One,\Two,\Two)=(1,0)$.
\end{claim}

\begin{proof}
The admissible outcomes of $(B,\One,\Two,\Two)$ are $(0,2),(0,3),(1,0),(1,2),(1,3)$. Indeed, the optimal outcome is $(1,0)$ with a social cost of $3$ and any outcome not listed has social cost at least $5$. We will show that for each of these outcomes besides $(1,0)$, there is a profitable deviation, contradicting that $M$ is strategyproof.
\begin{itemize}
\item $(0,2)$ and $(0,3)$. 
If the truthful profile is $(B,\One,\Two,\Two)$, the leaf-$1$ agent has true preference $\One$ and cost $1$. 
By misreporting $B$, she can change to the anchor profile $(B,B,\Two,\Two)$ and the outcome becomes $(1,0)$ due to Lemma~\ref{lem:anchors}. Hence, she can decrease her cost to $0$. 

\item $(1,2)$ and $(1,3)$. 
If the truthful profile is $(B,\One,\Two,\Two)$, the central agent has true preference $B$ and cost $2$.
By misreporting $\Two$, she can change to the anchor profile $(\Two,\One,\Two,\Two)$ and the outcome becomes $(1,0)$ due to Lemma~\ref{lem:anchors}. Hence, she can decrease her cost to $1$.
\end{itemize}
The only remaining admissible outcome is $(1,0)$.
\end{proof}

\begin{claim}\label{cl:B221}
$M(B,\Two,\Two,\One)=(3,0)$.
\end{claim}

\begin{proof}
The admissible outcomes of $(B,\Two,\Two,\One)$ are $(0,1),(0,2),(3,0),(3,1),(3,2)$. Indeed, the optimal outcome is $(3,0)$ with a social cost of $3$, and any outcome not listed has social cost at least $5$. We will show that for each outcome besides $(3,0)$, there is a profitable deviation, contradicting that $M$ is strategyproof.
\begin{itemize}
\item $(0,1)$ and $(0,2)$.
If the truthful profile is $(B,\Two,\Two,\One)$, the leaf-$3$ agent has true preference $\One$ and cost $1$. 
By misreporting $B$, she can change to the anchor profile $(B,\Two,\Two,B)$ and the outcome becomes $(3,0)$ due to Lemma~\ref{lem:anchors}. Hence, she can decrease her cost to $0$.

\item $(3,1)$ and $(3,2)$.
If the truthful profile is $(B,\Two,\Two,\One)$, the central agent has true preference $B$ and cost $2$. 
By misreporting $\Two$, she can change to the anchor profile $(\Two,\Two,\Two,\One)$ and the outcome becomes $(3,0)$ due to Lemma~\ref{lem:anchors}. Hence, she can decrease her cost to $1$. 
\end{itemize}
The only remaining admissible outcome is $(3,0)$.
\end{proof}

\begin{claim}\label{cl:B2BB}
$M(B,\Two,B,B)=(0,1)$.
\end{claim}

\begin{proof}
The admissible outcomes of $(B,\Two,B,B)$ are $(0,1)$, $(0,2)$, $(0,3)$, $(1,0)$, $(2,0)$, $(2,1)$, $(2,3)$, $(3,0)$, $(3,1)$, $(3,2)$. Indeed, the optimal outcome has social cost $6$ and the only two outcomes not listed have social cost $10$.
We will show that for each of these outcomes besides $(0,1)$, there is a profitable deviation. 
\begin{itemize}
\item $(0,2)$.
If the anchor profile $(B,\Two,\Two,B)$ is the truthful one, then the outcome is $(3,0)$ due to Lemma~\ref{lem:anchors}, and the leaf-$2$ agent with true preference $\Two$ has cost $1$. By misreporting $B$, she can change the profile to $(B,\Two,B,B)$ and the outcome becomes $(0,2)$. Hence, she can decrease her cost to $0$.

\item $(0,3)$.
If the anchor profile $(B,\Two,B,\Two)$ is the truthful one, then the outcome is $(2,0)$ due to Lemma~\ref{lem:anchors}, and the leaf-$3$ agent with true preference $\Two$ has cost $1$. By misreporting $B$, she can change the profile to $(B,\Two,B,B)$ and the outcome becomes $(0,3)$. Hence, she can decrease her cost to $0$.

\item $(1,0)$, $(2,0)$ and $(3,0)$.
If the profile $(B,\Two,B,B)$ is the truthful one and the outcome is one of $(1,0)$, $(2,0)$ or $(3,0)$, then the leaf-$1$ agent with true preference $\Two$ has cost $1$. By misreporting $B$, she can change the profile to $P_0=(B,B,B,B)$, whose outcome is $(0,1)$ by the standing assumption. Hence, she can decrease her cost to $0$.

\item $(2,1)$.
If the truthful profile is $(B,\Two,B,B)$, then the leaf-$3$ agent with true preference $B$ has cost $4$. 
By misreporting $\Two$, she can change the profile to $(B,\Two,B,\Two)$ and the outcome becomes $(2,0)$ due to Lemma~\ref{lem:anchors}. 
Hence, she can decrease her cost to $3$.

\item $(2,3)$.
If the anchor profile $(B,\Two,B,\Two)$ is the truthful one, then the outcome is $(2,0)$ due to Lemma~\ref{lem:anchors}, and the leaf-$3$ agent with true preference $\Two$ has cost $1$. By misreporting $B$, she can change the profile to $(B,\Two,B,B)$ and the outcome becomes $(2,3)$. Hence, she can decrease her cost to $0$.

\item $(3,1)$.
If the truthful profile is $(B,\Two,B,B)$, then the leaf-$2$ agent with true preference $B$ has cost $4$. 
By misreporting $\Two$, she can change the profile to $(B,\Two,\Two,B)$ and the outcome becomes $(3,0)$ due to Lemma~\ref{lem:anchors}. 
Hence, she can decrease her cost to $3$.

\item $(3,2)$.
If the anchor profile $(B,\Two,\Two,B)$ is the truthful one, then the outcome is $(3,0)$ due to Lemma~\ref{lem:anchors}, and the leaf-$2$ agent with true preference $\Two$ has cost $1$. By misreporting $B$, she can change the profile to $(B,\Two,B,B)$ and the outcome becomes $(3,2)$. Hence, she can decrease her cost to $0$.
\end{itemize}
The only remaining admissible outcome is $(0,1)$.
\end{proof}

\begin{claim}\label{cl:B2B1}
$M(B,\Two,B,\One)=(0,1)$.
\end{claim}

\begin{proof}
The admissible outcomes of $(B,\Two,B,\One)$ are $(0,1)$, $(0,2)$, $(0,3)$, $(1,0)$, $(2,0)$, $(2,1)$, $(3,0)$, $(3,1)$, $(3,2)$. 
Indeed, the optimal outcome is $(0,1)$ with social cost $5$, and any outcome not listed has social cost at least $8$.
We will show that for each of these outcomes besides $(0,1)$, there is a profitable deviation.
\begin{itemize}
\item $(0,2)$.
If the truthful profile is $(B,\Two,\Two,\One)$, then the outcome is $(3,0)$ due to Claim~\ref{cl:B221}, and the leaf-$2$ agent with true preference $\Two$ has cost $1$. By misreporting $B$, she can change the profile to $(B,\Two,B,\One)$ and the outcome becomes $(0,2)$. Hence, she can decrease her cost to $0$.

\item $(0,3)$.
If the anchor profile $(B,\Two,B,\Two)$ is the truthful one, then the outcome is $(2,0)$ due to Lemma~\ref{lem:anchors}, and the leaf-$3$ agent with true preference $\Two$ has cost $1$. By misreporting $\One$, she can change the profile to $(B,\Two,B,\One)$ and the outcome becomes $(0,3)$. Hence, she can decrease her cost to $0$.

\item $(1,0)$.
If the anchor profile $(B,\One,B,\One)$ is the truthful one, then the outcome is $(0,2)$ due to Lemma~\ref{lem:anchors}, and the leaf-$1$ agent with true preference $\One$ has cost $1$. By misreporting $\Two$, she can change the profile to $(B,\Two,B,\One)$ and the outcome becomes $(1,0)$. Hence, she can decrease her cost to $0$.

\item $(2,0)$ and $(2,1)$.
If the truthful profile is $(B,\Two,B,\One)$ and the outcome is either $(2,0)$ or $(2,1)$, then the leaf-$3$ agent with true preference $\One$ has cost $2$. By misreporting $B$, she can change the profile to $(B,\Two,B,B)$ and the outcome becomes $(0,1)$ due to Claim~\ref{cl:B2BB}. Under this outcome, her true cost is $1$. Hence, she can decrease her cost.

\item $(3,0)$.
If the truthful profile is $(B,\Two,B,B)$, then the outcome is $(0,1)$ due to Claim~\ref{cl:B2BB}, and the leaf-$3$ agent with true preference $B$ has cost $3$. By misreporting $\One$, she can change the profile to $(B,\Two,B,\One)$ and the outcome becomes $(3,0)$. 
Hence, she can decrease her cost to $1$.

\item $(3,1)$.
If the truthful profile is $(B,\Two,B,\One)$, then the leaf-$2$ agent with true preference $B$ has cost $4$.
By misreporting $\Two$, she can change the profile to $(B,\Two,\Two,\One)$ and the outcome becomes $(3,0)$ due to Claim~\ref{cl:B221}. 
Hence, she can decrease her cost to $3$.

\item $(3,2)$.
If the truthful profile is $(B,\Two,\Two,\One)$, then the outcome is $(3,0)$ due to Claim~\ref{cl:B221}, and the leaf-$2$ agent with true preference $\Two$ has cost $1$. By misreporting $B$, she can change the profile to $(B,\Two,B,\One)$ and the outcome becomes $(3,2)$. Hence, she can decrease her cost to $0$.
\end{itemize}
The only remaining admissible outcome is $(0,1)$.
\end{proof}

\begin{claim}\label{cl:BBBB1}
$M(B,B,B,\One)=(0,1)$.
\end{claim}

\begin{proof}
The admissible outcomes of $(B,B,B,\One)$ are $(0,1)$, $(0,2)$, $(0,3)$, $(1,0)$, $(1,2)$, $(2,0)$, $(2,1)$, $(3,0)$, $(3,1)$, $(3,2)$. 
Indeed, the optimal social cost is $6$ and the missing outcomes have social cost at least $9$.
We will show that for each of these outcomes besides $(0,1)$, there is a profitable deviation.
\begin{itemize}
\item $(0,2)$.
If the truthful profile is $(B,B,B,\One)$ and the outcome is $(0,2)$, then the leaf-$1$ agent with true preference $B$ has cost
$3$. By misreporting $\Two$, she can change the profile to $(B,\Two,B,\One)$, whose outcome is $(0,1)$ due to Claim~\ref{cl:B2B1}. 
Hence, she can decrease her cost to $1$.

\item $(0,3)$ and $(3,0)$.
If the truthful profile is $P_0=(B,B,B,B)$, then the outcome is $(0,1)$ by the standing assumption, and the leaf-$3$ agent with true preference $B$ has cost $3$.
By misreporting $\One$, she can change the profile to $(B,B,B,\One)$ and the outcome becomes either $(0,3)$ or $(3,0)$. 
Hence, she can decrease her cost to $1$.

\item $(1,0)$ and $(1,2)$.
If the anchor profile $(B,\One,B,\One)$ is the truthful one, then the outcome is $(0,2)$ due to Lemma~\ref{lem:anchors}, and the leaf-$1$ agent with true preference $\One$ has cost $1$. By misreporting $B$, she can change the profile to $(B,B,B,\One)$ and the outcome becomes either $(1,0)$ or $(1,2)$. Hence, she can decrease her cost to $0$.

\item $(2,0)$ and $(2,1)$.
If the anchor profile $(B,B,\One,\One)$ is the truthful one, then the outcome is $(0,1)$ due to Lemma~\ref{lem:anchors}, and the leaf-$2$ agent with true preference $\One$ has cost $1$. By misreporting $B$, she can change the profile to $(B,B,B,\One)$ and the outcome becomes either $(2,0)$ or $(2,1)$. Hence, she can decrease her cost to $0$.

\item $(3,1)$.
If the truthful profile is $(B,B,B,\One)$ and the outcome is $(3,1)$, then the leaf-$2$ agent with true preference $B$ has cost $4$.
By misreporting $\One$, she can change the profile to $(B,B,\One,\One)$ and the outcome becomes $(0,1)$ due to Lemma~\ref{lem:anchors}. 
Hence, she can decrease her cost to $3$.

\item $(3,2)$.
If the truthful profile is $(B,B,B,\One)$ and the outcome is $(3,2)$, then the leaf-$1$ agent with true preference $B$ has cost $4$.
By misreporting $\One$, she can change the profile to $(B,\One,B,\One)$ and the outcome becomes $(0,2)$ due to Lemma~\ref{lem:anchors}. 
Hence, she can decrease her cost to $3$.
\end{itemize}
The only remaining admissible outcome is $(0,1)$.
\end{proof}

\begin{claim}\label{cl:BB21}
$M(B,B,\Two,\One)=(3,0)$.
\end{claim}

\begin{proof}
The admissible outcomes of $(B,B,\Two,\One)$ are $(0,1)$, $(0,2)$, $(0,3)$, $(1,0)$, $(1,2)$, $(2,0)$, $(3,0)$, $(3,1)$, $(3,2)$. 
Indeed, the optimal outcome is $(3,0)$ with social cost $5$, and any missing outcome has social cost at least $8$.
We will show that for each of these outcomes besides $(3,0)$, there is a profitable deviation.
\begin{itemize}
\item $(0,1)$.
If the truthful profile is $(B,\Two,\Two,\One)$, then the outcome is $(3,0)$ due to Claim~\ref{cl:B221}, and the leaf-$1$ agent with true preference $\Two$ has cost $1$. By misreporting $B$, she can change the profile to $(B,B,\Two,\One)$ and the outcome becomes $(0,1)$. Hence, she can decrease her cost to $0$.

\item $(0,2)$.
If the truthful profile is $(B,B,B,\One)$, then the outcome is $(0,1)$ due to Claim~\ref{cl:BBBB1}, and the leaf-$2$ agent with true preference $B$ has cost $3$. By misreporting $\Two$, she can change the profile to $(B,B,\Two,\One)$ and the outcome becomes $(0,2)$.
Hence, she can decrease her cost to $1$.

\item $(0,3)$.
If the anchor profile $(B,B,\Two,\Two)$ is the truthful one, then the outcome is $(1,0)$ due to Lemma~\ref{lem:anchors}, and the leaf-$3$ agent with true preference $\Two$ has cost $1$. By misreporting $\One$, she can change the profile to $(B,B,\Two,\One)$ and the outcome becomes $(0,3)$. Hence, she can decrease her cost to $0$.

\item $(1,0)$ and $(1,2)$.
If the truthful profile is $(B,\One,\Two,\One)$, then the outcome is $(0,2)$ due to Claim~\ref{cl:B121}, and the leaf-$1$ agent with true preference $\One$ has cost $1$. By misreporting $B$, she can change the profile to $(B,B,\Two,\One)$ and the outcome becomes either $(1,0)$ or $(1,2)$. Hence, she can decrease her cost to $0$.

\item $(2,0)$.
If the anchor profile $(B,B,\One,\One)$ is the truthful one, then the outcome is $(0,1)$ due to Lemma~\ref{lem:anchors}, and the leaf-$2$ agent with true preference $\One$ has cost $1$. By misreporting $\Two$, she can change the profile to $(B,B,\Two,\One)$ and the outcome becomes $(2,0)$. Hence, she can decrease her cost to $0$. 

\item $(3,1)$.
If the truthful profile is $(B,\Two,\Two,\One)$, then the outcome is $(3,0)$ due to Claim~\ref{cl:B221}, and the leaf-$1$ agent with true preference $\Two$ has cost $1$. By misreporting $B$, she can change the profile to $(B,B,\Two,\One)$ and the outcome becomes $(3,1)$. 
Hence, she can decrease her cost to $0$.

\item $(3,2)$.
If the profile $(B,B,\Two,\One)$ is the truthful one and the outcome is $(3,2)$, then the leaf-$1$ agent with true preference $B$ has cost $4$.
By misreporting $\One$, she can change the profile to $(B,\One,\Two,\One)$, whose outcome is $(0,2)$ due to Claim~\ref{cl:B121}. 
Hence, she can decrease her cost to $3$.
\end{itemize}
The only remaining admissible outcome is $(3,0)$.
\end{proof}

\begin{claim}\label{cl:BB2B}
$M(B,B,\Two,B)=(3,0)$.
\end{claim}

\begin{proof}
The admissible outcomes of $(B,B,\Two,B)$ are $(0,1)$, $(0,2)$, $(0,3)$, $(1,0)$, $(1,2)$, $(1,3)$, $(2,0)$, $(3,0)$, $(3,1)$, $(3,2)$,
as the optimal cost, due to $(3,0)$, is $6$ and all missing outcomes have cost $10$.
We will show that for each of these outcomes besides $(3,0)$, there is a profitable deviation.
\begin{itemize}
\item $(0,1)$.
If the anchor profile $(B,\Two,\Two,B)$ is the truthful one, then the outcome is $(3,0)$ due to Lemma~\ref{lem:anchors}, and the leaf-$1$ agent with true preference $\Two$ has cost $1$. By misreporting $B$, she can change the profile to $(B,B,\Two,B)$ and the outcome becomes $(0,1)$. Hence, she can decrease her cost to $0$.

\item $(0,2)$ and $(2,0)$.
If the truthful profile is $P_0=(B,B,B,B)$, then the outcome is $(0,1)$ by the standing assumption, and the leaf-$2$ agent with true preference $B$ has cost $3$. By misreporting $\Two$, she can change the profile to $(B,B,\Two,B)$ and the outcome becomes either $(0,2)$ or $(2,0)$. 
Hence, she can decrease her cost to $1$.

\item $(0,3)$ and $(1,3)$.
If the anchor profile $(B,B,\Two,\Two)$ is the truthful one, then the outcome is $(1,0)$ due to Lemma~\ref{lem:anchors}, and the leaf-$3$ agent with true preference $\Two$ has cost $1$. By misreporting $B$, she can change the profile to $(B,B,\Two,B)$ and the outcome becomes either $(0,3)$ or $(1,3)$. Hence, she can decrease her cost to $0$.

\item $(1,0)$.
If the truthful profile is $(B,B,\Two,B)$ and the outcome is $(1,0)$, then the leaf-$3$ agent with true preference $B$ has cost $3$.
By misreporting $\One$, she can change the profile to $(B,B,\Two,\One)$ and the outcome becomes $(3,0)$ due to Claim~\ref{cl:BB21}. 
Hence, she can decrease her cost to $1$.

\item $(1,2)$.
If the truthful profile is $(B,B,\Two,B)$ and the outcome is $(1,2)$, then the leaf-$3$ agent with true preference $B$ has cost $4$.
By misreporting $\Two$, she can change the profile to $(B,B,\Two,\Two)$ and the outcome becomes $(1,0)$ due to Lemma~\ref{lem:anchors}. 
Hence, she can decrease her cost to $3$.

\item $(3,1)$.
If the anchor profile $(B,\Two,\Two,B)$ is the truthful one, then the outcome is $(3,0)$ due to Lemma~\ref{lem:anchors}, and the leaf-$1$ agent with true preference $\Two$ has cost $1$. By misreporting $B$, she can change the profile to $(B,B,\Two,B)$ and the outcome becomes $(3,1)$. Hence, she can decrease her cost to $0$.

\item $(3,2)$.
If the truthful profile is $(B,B,\Two,B)$ and the outcome is $(3,2)$, then the leaf-$1$ agent with true preference $B$ has cost $4$.
By misreporting $\Two$, she can change the profile to $(B,\Two,\Two,B)$, whose outcome is $(3,0)$ due to Lemma~\ref{lem:anchors}. 
Hence, she can decrease her cost to $3$.
\end{itemize}
The only remaining admissible outcome is $(3,0)$.
\end{proof}

We are now finally ready to prove that $M(B,\One,\Two,B)=(0,2)$. 
The admissible outcomes for this profile are $(0,1)$, $(0,2)$, $(0,3)$, $(1,0)$, $(1,2)$, $(1,3)$, $(2,0)$, $(3,0)$, $(3,2)$ since the optimal outcome is $(0,2)$ with social cost $5$ and the missing outcomes have social cost at least $8$. We argue that for any of these outcomes besides $(0,2)$, there is a profitable deviation. 
\begin{itemize}
\item $(0,1)$.
If the truthful profile is $(B,\Two,\Two,B)$, then the outcome is $(3,0)$ due to Lemma~\ref{lem:anchors}, and the leaf-$1$ agent with true preference $\Two$ has cost $1$. By misreporting $\One$, she can change the profile to $(B,\One,\Two,B)$ and the outcome becomes $(0,1)$. 
Hence, she can decrease her cost to $0$. 

\item $(1,0)$.
If the truthful profile is $(B,B,\Two,B)$, then the outcome is $(3,0)$ due to Claim~\ref{cl:BB2B}, and the leaf-$1$ agent with true preference $B$ has cost $3$.  By misreporting $\One$, she can change the profile to $(B,\One,\Two,B)$ and the outcome becomes $(1,0)$. 
Hence, she can decrease her cost to $1$.

\item $(0,3)$ and $(1,3)$.
If the truthful profile is $(B,\One,\Two,\Two)$, then the outcome is $(0,2)$ due to Claim~\ref{cl:B122}, and the leaf-$3$ agent with true preference $\Two$ has cost $1$. By misreporting $B$, she can change the profile to $(B,\One,\Two,B)$ and the outcome becomes either $(0,3)$ or $(1,3)$. Hence, she can decrease her cost to $0$.

\item $(1,2)$.
If the truthful profile is $(B,\One,\Two,B)$ and the outcome is $(1,2)$, then the leaf-$3$ agent with true preference $B$ has cost $4$.
By misreporting $\One$, she can change the profile to $(B,\One,\Two,\One)$ and the outcome becomes $(0,2)$ due to Claim~\ref{cl:B121}. 
Hence, she can decrease her cost to $3$.

\item $(2,0)$.
If the truthful profile is $(B,\One,\One,B)$, then the outcome is $(0,3)$ due to Lemma~\ref{lem:anchors}, and the leaf-$2$ agent with true preference $\One$ has cost $1$. By misreporting $\Two$, she can change the profile to $(B,\One,\Two,B)$ and the outcome becomes $(2,0)$. Hence, she can decrease her cost to $0$.

\item $(3,0)$ and $(3,2)$.
If the truthful profile is $(B,\One,\Two,\One)$, then the outcome is $(0,2)$ due to Claim~\ref{cl:B121}, and the leaf-$3$ agent with true preference $\One$ has cost $1$. By misreporting $B$, she can change the profile to $(B,\One,\Two,B)$ and the outcome becomes either $(3,0)$ or $(3,2)$. Hence, she can decrease her cost to $0$.
\end{itemize}
The only remaining admissible outcome is $(0,2)$.
\hfill $\qed$

\end{document}